%% file: paper.tex
\documentclass{sig-alternate}

\newfont{\mycrnotice}{ptmr8t at 7pt}
\newfont{\myconfname}{ptmri8t at 7pt}
\let\crnotice\mycrnotice%
\let\confname\myconfname%

\permission{}
\conferenceinfo{To appear at PODS'15,}{May 31--June 4, 2015, Melbourne, Victoria, Australia.}
\copyrightetc{}
\crdata{http://dx.doi.org/10.1145/2745754.2745760}

\usepackage{microtype}%if unwanted, comment out or use option "draft"

\usepackage{color}      % use if color is used in text
\usepackage{xspace}
\usepackage{epsfig} 
\usepackage{hyperref}   % use for hypertext links, including those to external documents and URLs
\usepackage{enumerate}

\usepackage{amsthm}

\input{macros}

\title{Symmetric Weighted First-Order Model Counting}
\numberofauthors{2} 
\author{
\alignauthor
Paul Beame \\
\affaddr{University of Washington}\\
\email{beame@cs.washington.edu}
\alignauthor
Guy Van den Broeck \\
\affaddr{KU Leuven}\\
\email{guy.vandenbroeck@cs.kuleuven.be}
\and
\alignauthor Eric Gribkoff\\
\affaddr{University of Washington}\\
\email{eagribko@cs.washington.edu}
\alignauthor Dan Suciu\\ 
\affaddr{University of Washington}\\
\email{suciu@cs.washington.edu}
}

\begin{document}

% PODS requires the insertion of these two lines:
\fontsize{10.1pt}{10.3pt} 
%\fontsize{10.2pt}{10.4pt} % (few) overfull hbox errors, but a couple
                           %references over the page limit
\selectfont
%\overfullrule=2cm % show overfull hbox
% explanation from the web:  \fontsize{<size>}{<line space>}

\maketitle

\begin{abstract}
  The FO Model Counting problem (FOMC) is the following: given a
  sentence $\Phi$ in FO and a number $n$, compute the number of models
  of $\Phi$ over a domain of size $n$; the Weighted variant (WFOMC)
  generalizes the problem by associating a weight to each tuple and
  defining the weight of a model to be the product of weights of its
  tuples.  In this paper we study the complexity of the symmetric
  WFOMC, where all tuples of a given relation have the same weight.
  Our motivation comes from an important application, inference in
  Knowledge Bases with soft constraints, like Markov Logic Networks,
  but the problem is also of independent theoretical interest.  We
  study both the data complexity, and the combined complexity of FOMC
  and WFOMC.  For the data complexity we prove the existence of an
  \FO{3} formula for which FOMC is \SHARPONE-complete, and the
  existence of a Conjunctive Query for which WFOMC is
  \SHARPONE-complete.  We also prove that all $\gamma$-acyclic queries
  have polynomial time data complexity.  For the combined complexity,
  we prove that, for every fragment \FO{k}, $k\geq 2$, the combined
  complexity of FOMC (or WFOMC) is \SHARP-complete.
%
%   This problem occurs in many settings, such as 0-1 laws, data
%   privacy, and Statistical Relational Models (such as Markov Logic
%   Networks), and it was observed repeatedly that for some formulas
%   this quantity is very difficult to compute, yet no formal hardness
%   proof has been given in the past.  In fact, this problem is unlikely
%   to be \SHARP-hard because here the input alphabet is unary; the
%   symmetric WFOMC belongs to the class \SHARPONE, which is the set of
%   \SHARP{} problems over a unary input alphabet.  We present two
%   results.  First, we show that there exists a formula $\HardSentence$
%   in \FO{3} for which the symmetric WFOMC problem is \SHARPONE-hard;
%   this also implies the existence of a conjunctive query that is
%   \SHARPONE-hard.  Second, we prove several complexity results for
%   \FO{2} sentences and for $\gamma$-acyclic conjunctive queries: in
%   particular, both have polynomial time data complexity.
  % The Weighted Model Counting (WMC) problem generalizes this by
  % assigning a weight $w(X)$ to each Boolean variable $X$ and asking
  % for the sum of weights of all models satisfying $F$, $WMC(F,w)$,
  % where the weight of a model is the product of weights of variables
  % set to true. In this paper we study the following variant of the
  % WFOMC problem, where the weights of the tuples depend only on the
  % relation name. We call this problem the Symmetric WFOMC.
\end{abstract}

% \guy{ 15 pages excluding references}

\input{intro}
\input{background}
\input{data-complexity}
\input{combined-complexity}

\input{conclusions}

\bibliographystyle{plain}
\bibliography{bib}  

\clearpage

% \onecolumn

\appendix

\input{sharpp1.tex}

\end{document}

%% file: macros.tex
\newcommand{\eqdef}{\stackrel{\text{def}}{=}}

\newcommand{\cut}[1]{}
\newcommand{\commentresolved}[1]{}

\newcommand{\eg}{{\em e.g.}\xspace}

\newcommand{\set}[1]{\{#1\}}                    % Set (as in \set{1,2,3}).
\newcommand{\setof}[2]{\{{#1}\mid{#2}\}}        % Set (as in \setof{x}{x>0}).
\usepackage{aliascnt}  		% ``hyperref’s \autoref command does not work well with theorems that share a counter:
 \newtheorem{theorem}{Theorem}[section]          	% Theorem environment.
 \newaliascnt{lemma}{theorem}				% 1 alias counter
 \newtheorem{lemma}[lemma]{Lemma}              	% Lemma environment.
 \aliascntresetthe{lemma}  					% 3 set
 \newaliascnt{conjecture}{theorem}			% 1 alias counter
 \aliascntresetthe{conjecture}  				% 3 set
 \newaliascnt{remark}{theorem}				% 1 alias counter
               
 \aliascntresetthe{remark}  					% 3 set
 \newaliascnt{corollary}{theorem}			% 1 alias counter
 \newtheorem{corollary}[corollary]{Corollary}      % Corollary environment.
 \aliascntresetthe{corollary}  				% 3 set
 \newaliascnt{definition}{theorem}			% 1 alias counter
 \aliascntresetthe{definition}  				% 3 set
 \newaliascnt{proposition}{theorem}			% 1 alias counter
 \aliascntresetthe{proposition}  				% 3 set
 \newaliascnt{example}{theorem}			% 1 alias counter
 \newtheorem{example}[example]{Example}  	% 2 environment.
 \aliascntresetthe{example}  				% 3 set
\def\sp1{\#P$_1$}

\def\Min{\mathit{Min}}
\def\Max{\mathit{Max}}
\def\Succ{\mathit{Succ}}
\def\GoLeft{\mathit{Left}}

\def\GoRight{\mathit{Right}}
\def\Unchanged{\mathit{Unchanged}}

\DeclareMathOperator{\WFOMC}{WFOMC}
\DeclareMathOperator{\FOMC}{FOMC}
\DeclareMathOperator{\UFOMC}{UFOMC}
\DeclareMathOperator{\WMC}{WMC}
\DeclareMathOperator{\States}{States}
\DeclareMathOperator{\Tup}{Tup}
\newcommand{\FO}[1]{FO$^{#1}$}
\newcommand{\SHARP}{\mbox{\#P}}
\newcommand{\SHARPONE}{\mbox{\#P$_1$}}
\newcommand{\NPONE}{\mbox{NP$_1$}}

\newcommand{\HardSentence}{\Theta_1}
\newcommand{\HardClause}{\Xi_1}
\newcommand{\HardCQ}{\Upsilon_1}

\newcommand{\UTM}{U_1}
\newcommand{\DTM}{T_{\text{det}}}

\newcommand\ignore[1]{\relax}

\newcommand\summary[2]{#2}

%% file: intro.tex
\section{Introduction}

\label{sec:intro}
% 
% The Model Counting (MC) Problem asks for the number of satisfying
% assignments, $\#F$, of some Boolean formula $F$.  The Weighted Model
% Counting (WMC) problem generalizes this by assigning a weight $w(X)$
% to each Boolean variable $X$ and asking for the sum of weights of all
% models satisfying $F$, $\WMC(F,w)$, where the weight of a model is the
% product of weights of variables set to true.  An equivalent problem is
% this: set each Boolean variable $X$ to true independently, with some
% known probability $p(X)$, and compute the probability that $F$ is
% true.  All these problems are in \#P, and model counting is already
% \#P-hard, even for some restricted classes of formulas where
% satisfiability is easy to check
% \cite{DBLP:journals/siamcomp/Valiant79}.
% %

%\guy{We should mention the term lifted inference, and liftability at least once, to bridge the terminology for AI people that want to read this.}
%\dan{done; please add a citation that you find appropriate}

\begin{sloppypar}
Probabilistic inference is becoming a central data management
problem. Large knowledge bases, such as
Yago~\cite{DBLP:journals/ai/HoffartSBW13},
Nell~\cite{DBLP:conf/aaai/CarlsonBKSHM10}, DeepDive~\cite{deepdive},
Reverb~\cite{DBLP:conf/emnlp/FaderSE11}, Microsoft's
Probase~\cite{DBLP:conf/sigmod/WuLWZ12} or Google's Knowledge
Vault~\cite{knoweldge-vault-kdd-2014}, have millions to billions of
uncertain tuples.  These systems scan large corpora of text, such as
the Web or complete collections of journal articles, and extract
automatically billions of structured facts, representing large
collections of knowledge.  For an illustration, Google's Knowledge
Vault~\cite{knoweldge-vault-kdd-2014} contains 1.6B triples of the
form (subject, predicate, object), for example,
\texttt{</m/02mjmr, /people/person/place\_of\_birth /m/02hrh0\_>}
%
%\verb+</m/02mjmr, /people/person/place_of_birth /m/02hrh0_>+, 
%$<\texttt{/m/02mjmr}, \texttt{/people/person/place\_of\_birth},$
%$\texttt{/m/02hrh0\_} >$
%
where
%
%\verb+/m/02mjmr+
\texttt{/m/02mjmr}
is the Freebase id for Barack Obama, and
%
%\verb+/m/02hrh0_+
\texttt{/m/02hrh0\_}
is the id for Honolulu~\cite{knoweldge-vault-kdd-2014}.  The triples
are extracted automatically from the Web, and each triple is annotated
with a probability $p$ representing the confidence in the extraction.
\end{sloppypar}

\begin{sloppypar}
A central and difficult problem in such systems is probabilistic
inference, or, equivalently weighted model counting.  The classical FO
Model Counting problem (FOMC) is: given a sentence $\Phi$ in First-Order 
Logic (FO) and a number $n$, compute the number of structures
over a domain of size $n$ that satisfy the sentence $\Phi$; in this
paper we consider only {\em labeled structures}, i.e.~isomorphic
structures are counted as distinct.  We denote the number of models by
$\FOMC(\Phi, n)$, for example $\FOMC(\forall x\exists y R(x,y), n) =
(2^n-1)^n$.\footnote{For a fixed $x$, there are $2^n$ assignments to $R(x,y)$, which all satisfy $\exists y R(x,y)$, except the one where all atoms are false. Moreover, the models for the $n$ values of $x$ can be counted independently and multiplied. }  In the Weighted FO Model Counting (WFOMC) variant, one
further associates a real number $w(t)$ called {\em weight} to each
tuple $t$ over the domain of size $n$, and defines the weight of a
structure as the product of the weights of all tuples in that
structure.  The Weighted Model Count $\WFOMC(\Phi, n, \mathbf{w})$ is
defined as the sum of the weights of all structures over a domain of
size $n$ that satisfy the sentence $\Phi$.  Weights map immediately to
probabilities, in the following way: if each tuple $t$ is included
in the database independently with probability $w(t)/(1+w(t))$, then
the probability that a formula $\Phi$ is true is $\Pr(\Phi)
=\WFOMC(\Phi, n, \mathbf{w})/\WFOMC(\text{true}, n, \mathbf{w})$,
where $\WFOMC(\text{true}, n, \mathbf{w}) = \prod_t (1+w(t))$ is the
sum of weights of all structures.
\end{sloppypar}

In this paper we study the {\em symmetric} WFMOC problem, where all
tuples from the same relation have the same weight, which
we denote $w_i$.  For example, a random graph $G(n,p)$ is a symmetric
structure, since every edge is present with the same probability $p$
(equivalently: has weight $p/(1-p)$), and FOMC is another special case
where all weights are set to 1.  The symmetric WFMOC problem occurs
naturally in Knowledge Bases with soft constraints, as we illustrate
next.

% 
% As
% a special case, consider the probability of a sentence $\Phi$ on an
% undirected random graph $G(n,p)$: it is given by $\Pr(\Phi) =
% (1-p)^{n(n-1)/2} \WFOMC(\Phi, n, p/(1-p))$ (notice that the weight of
% an edge is given by the odds $p/(1-p)$).  Another special case is the
% model counting problem: $\FOMC(\Phi, n) = \WFOMC(\Phi, n, 1)$.  In
% this paper we study the complexity of the symmetric WFOMC problem.  We
% consider both the data complexity, when $\Phi$ is fixed and the input
% consists of $n$, and the query or combined complexity, when $n$ is
% fixed and $\Phi$ is the input, or when both form the input.

\begin{example} \label{ex:1} A Markov Logic Network
  (MLN)~\cite{DBLP:series/synthesis/2009Domingos} is a finite set of
  soft or hard constraints.  Each constraint is a pair $(w, \varphi)$,
  where $\varphi$ is a formula, possibly with free variables
  $\mathbf{x}$, and $w \in [0,\infty]$ is a weight\footnote{In typical
    MLN systems, users specify the log of the weight rather than the
    weight.  The pair $(1.098, \varphi)$ means that the weight of $\varphi$ is $w = \text{exp}(1.098) \approx 3$.  Using
    logs simplifies the learning task. We do not address learning and will omit logs; $(w,\varphi)$ means that $\varphi$ has weight
    $w$.}.  For example, the soft constraint
\begin{align}
  (3, \text{Spouse}(x,y)\wedge \text{Female}(x) \Rightarrow \text{Male}(y)) \label{eq:mln}
\end{align}
specifies that, typically, a female's spouse is male, and associates
the weight $w=3$ to this constraint.  If $w = \infty$ then we call
$(w, \varphi)$ a hard constraint.

\begin{sloppypar}
The semantics of MLNs naturally extend the Weighted Model Counting
setting. Given a finite domain (set of constants), an MLN defines a 
probability distribution over all structures for that domain (also called {\em possible worlds}). 
Every structure $D$ has a
weight
\begin{align*}
  W(D) = \prod_{(w,\varphi(\mathbf{x})) \in \text{MLN},
      \mathbf{a}\in D^{|\mathbf{x}|}: w<\infty \wedge D\models
      \varphi[\mathbf{a}/\mathbf{x}]}w
\end{align*}
In other words, for each soft constraint $(W, \varphi)$, and for every
tuple of constants $\mathbf{a}$ such that $\varphi(\mathbf{a})$ holds
in $D$, we multiply $D$'s weight by $w$. For example, given the MLN
that consists only of the soft constraint (\ref{eq:mln}), the weight
of a world $D$ is $3^N$, where $N$ is the number of pairs of constants
$a,b$ for which $\text{Spouse}(a,b), \text{Female}(a) \Rightarrow
\text{Male}(b)$ holds in $D$. The weight $W(\Phi)$ of a sentence
$\Phi$ is defined as the sum of weights of all worlds $D$ that satisfy
both $\Phi$ and all hard constraints in the MLN; its probability is
obtained by normalizing $\Pr_{\mathit{MLN}}(\Phi) = W(\Phi)/W(\text{true})$.
Notice that the symmetric WFOMC problem corresponds to the special
case of an MLN consisting of one soft constraint $(w_i,
R_i(\mathbf{x}_i))$ for each relation $R_i$, where $|\mathbf{x}_i| =
\text{arity}(R_i)$.\end{sloppypar}
\end{example}

Today's MLN systems (Alchemy~\cite{alchemy},
Tuffy~\cite{DBLP:journals/pvldb/NiuRDS11,DBLP:conf/sigmod/ZhangR13})
use an MCMC algorithm called MC-SAT~\cite{DBLP:conf/aaai/PoonD06} for
probabilistic inference.  The theoretical convergence guarantees of
MC-SAT require access to a uniform sampler over satisfying assignments
to a set of constraints. In practice, MC-SAT implementations rely on
SampleSAT~\cite{DBLP:conf/aaai/WeiES04}, which provides no guarantees
on the uniformity of solutions. Several complex examples are known in
the literature where model counting based on SampleSAT leads to highly
inaccurate estimates ~\cite{gomes2007sampling}.

A totally different approach to computing $\Pr_{\mathit{MLN}}(\Phi)$ is to reduce
it to a symmetric
$\WFOMC$~\cite{van2011lifted,DBLP:conf/uai/GogateD11a,DBLP:conf/kr/BroeckMD14,DBLP:journals/pvldb/JhaS12},
and this motivates our current paper.  We review here briefly one such
reduction, adapting from
~\cite{DBLP:journals/pvldb/JhaS12,DBLP:conf/kr/BroeckMD14}. 

\begin{example} \label{ex:1:cont} 
\begin{sloppypar}
Given an MLN, replace every soft
  constraint $(w, \varphi(\mathbf{x}))$ by two new constraints:
  $(\infty, \forall \mathbf{x} (R(\mathbf{x}) \vee
  \varphi(\mathbf{x})))$ and $(1/(w-1), R(\mathbf{x}))$.  Here $R$ is
  a new relational symbol with the same arity as the number of free
  variables in $\varphi$, and the constraint $(1/(w-1),
  R(\mathbf{x}))$ defines $R$ as a relation where all tuples
  have weight $1/(w-1)$.  Therefore, the probability of a formula
  $\Phi$ in the MLN can be computed as a conditional probability over
  a symmetric, tuple-independent database: $\Pr_{\mathit{MLN}}(\Phi) =
  \Pr(\Phi | \Gamma)$, where $\Gamma$ is the conjunction of all hard
  constraints\footnote{The reason why this works is the following: in
    original MLN, each tuple $\mathbf{a}$ contributes to $W(D)$ a
    factor of 1 or $w$, depending on whether $\varphi(\mathbf{a})$ is
    false or true in $D$; after the rewriting, the contribution of
    $\mathbf{a}$ is $1/(w-1)$ when $\varphi(\mathbf{a})$ is false,
    because in that case $R(\mathbf{a})$ must be true, or $1+1/(w-1)=
    w/(w-1)$ when $\varphi(\mathbf{a})$ is true, because
    $R(\mathbf{a})$ can be either false or true. The ratio is the same
    $1:w = [1/(w-1)]:[w/(w-1)]$.  }.
    Note that this reduction to WFOMC is independent of the 
    finite domain under consideration.
\end{sloppypar}

For example, the soft constraint in (\ref{eq:mln}) is translated into 
the hard constraint:
\begin{align*}
  \forall x, y (R(x,y) \vee \neg \text{Spouse}(x,y) \vee \neg \text{Female}(x) \vee \text{Male}(y))
% \label{eq:hard}
\end{align*}
and a tuple-independent probabilistic relation $R$ where all tuples
have weight $1/(3-1) = 1/2$, or, equivalently, have probability
$(1/2)/(1+(1/2)) = 1/3$.
\end{example}

% The main motivation for symmetric WFOMC comes from Statistical
% Relational Models~\cite{DBLP:conf/ijcai/Poole03,van2011lifted}, which combine logic and
% probabilities for complex machine learning tasks.  In those
% frameworks, both learning and inference tasks are specified through
% formulas with weights.  The formulas have logical variables that refer
% to values in the domain, but the domain is only specified through the
% number $n$, thus there is no ``database''.  For example, in Markov
% Logic Networks~\cite{DBLP:series/synthesis/2009Domingos} (as implemented in
% Tuffy~\cite{DBLP:journals/pvldb/NiuRDS11}), one can write a soft
% constraint like $(3.75, \text{Follows}(x,y) \Rightarrow
% \text{Popular}(y))$, meaning that if $y$ has followers, then $y$ is
% likely to be popular (where $3.75$ is the weight of the rule).  The
% larger the number of followers $x$, the stronger the evidence that $y$
% is popular.  Another application comes from data
% privacy~\cite{DBLP:conf/icdt/DalviMS05}, where the random database
% represents the adversary's prior and is also given as a symmetric
% structure; the leakage about a query $Q_1$ due to publishing the
% answer to another query $Q_2$ can be computed as $P(Q_1 | Q_2)$.

Thus, our main motivation for studying the symmetric $\WFOMC$ is very
practical, as symmetric models have been extensively researched in the
AI community recently, for inference in MLNs and
beyond~\cite{kersting2009counting,van2011lifted,DBLP:conf/aaai/Niepert13,vihav-nips-2014}.
Some tasks on MLNs, such as parameter learning~\cite{VdBStarAI13b}, naturally exhibit symmetries. For others, such as computing conditional probabilities given a large ``evidence'' database, the symmetric $\WFOMC$ model is applicable when the database has bounded Boolean rank~\cite{VdBNIPS13}.
Moreover, the problem is of independent theoretical interest as we
explain below.  We study both the data complexity, and the combined
complexity.  In both settings we assume that the vocabulary $\sigma =
(R_1, \ldots, R_m)$ is fixed, and so are the weights
$\mathbf{w} = (w_1, \ldots, w_m)$ associated with the relations.  In {\em data complexity}, the
formula $\Phi$ is fixed, and the only input is the number $n$
representing the size of the domain.  In this case $\WFOMC$ is a
counting problem over a unary alphabet: {\em given an input $1^n$,
  compute $\WFOMC(\Phi, n, \mathbf{w})$}.  It is immediate that this
problem belongs to the class \SHARPONE, which is the set of \#P
problems over a unary input
alphabet~\cite{DBLP:journals/siamcomp/Valiant79}.  In the {\em
  combined complexity}, both $n$ and the formula $\Phi$ are input.

In this paper we present results on the data complexity and the
combined complexity of the FOMC and WFOMC problem, and also some
results on the associated decision problem.

\subsubsection*{Results on Data Complexity}
In a surprising result~\cite{DBLP:conf/kr/BroeckMD14} has proven that
for \FO{2} the data complexity of symmetric $\WFOMC$ is in
PTIME (reviewed in Appendix~\ref{s:fo2-ptime-proof}).\footnote{PTIME data complexity for symmetric $\WFOMC$ is called \emph{domain-liftability} in the AI and lifted inference literature~\cite{broeck2011completeness}.}  This is
surprising because \FO{2} (the class of FO formulas restricted to two
logical variables) contains many formulas for which the asymmetric
problem was known to be \SHARP-hard.  An example is $\Phi=\exists x
\exists y (R(x) \wedge S(x,y) \wedge T(y))$, which is \SHARP-hard over
asymmetric structures,
but the number of models 
is\footnote{Fix the relations $R, T$, and let their cardinalities be
  $|R|=k$ and $|T|=m$. Then the structure does {\em not} satisfy
  $\Phi$ iff $S$ contains none of the $km$ tuples in $R \times T$, proving the formula.}  $2^{2n+n^2} - \sum_{k,m}{n \choose k}{n
  \choose m}2^{n^2-km}$, which is a number computable in time
polynomial in $n$.\footnote{Tractability of $\Phi$ was noted before in, for example~\cite{DBLP:series/synthesis/2011Suciu,broeck2011completeness}.} More generally, the symmetric WFOMC problem for
$\Phi$ is in PTIME.

This begs the question: could it be the case that {\em every} FO
formula is in PTIME?  The answer was shown to be negative by Jaeger
and Van den Broeck~\cite{jaeger-broeck-2012,jaegerlower}, using the
following argument.  Recall that the {\em spectrum},
$\texttt{Spec}(\Phi)$, of a formula $\Phi$ is the set of numbers $n$
for which $\Phi$ has a model over a domain of size
$n$~\cite{DBLP:journals/bsl/DurandJMM12}.  Jaeger and Van den Broeck
observed that the spectrum membership problem, ``{\em is $n \in
  \texttt{Spec}(\Phi)$?}'', can be reduced to WFOMC, by checking whether
$\FOMC(\Phi,n) > 0$.  Then, using a result in~\cite{jones1972turing}, if $\texttt{ETIME} \neq \texttt{NETIME}$, then there
exists a formula $\Phi$ for which computing $\WFOMC$ is not in
polynomial time\footnote{Recall that $\texttt{ETIME} = \bigcup_{c \geq
    1} \texttt{DTIME}(2^{cn})$ and $\texttt{NETIME} = \bigcup_{c \geq
    1} \texttt{NTIME}(2^{cn})$, and are not to be confused with the
  more familiar classes EXPTIME and NEXPTIME, which are $\bigcup_{c
    \geq 1} \texttt{DTIME}(2^{n^c})$ and $\bigcup_{c \geq 1}
  \texttt{NTIME}(2^{n^c})$ respectively.}. However, no hardness
results for the symmetric $\WFOMC$ were known to date.

What makes the data complexity of the symmetric WFOMC difficult to
analyze is the fact that the input is a single number $n$.  Valiant
already observed in~\cite{DBLP:journals/siamcomp/Valiant79} that such
problems are probably not candidates for being \#P-complete.  Instead,
he defined the complexity class \SHARPONE, to be the set of counting
problems for NP computations over a single-letter input alphabet.
Very few hardness results are known for this class: we are aware only
of a graph matching problem that was proven by Valiant, and of a
language-theoretic problem by Bertoni and
Goldwurm~\cite{DBLP:journals/ita/BertoniG93}.

Our data complexity results are the following. First, we establish the
existence of an FO sentence $\HardSentence$ for which the data
complexity of the $\FOMC$ problem is \SHARPONE-hard; and we also
establish the existence of a conjunctive query $\HardCQ$ for which the
data complexity of the $\WFOMC$ problem is \SHARPONE-hard.  Second, we
prove that every $\gamma$-acyclic conjunctive query without self-joins
is in polynomial time, extending the result
in~\cite{DBLP:conf/kr/BroeckMD14} from \FO{2} to $\gamma$-acyclic
conjunctive queries.  We give now more details about our results, and
explain their significance.

%Ours is, to the best of our knowledge, the first result to establish
%that the symmetric WFOMC problem is \SHARPONE-hard for any given query.  As we
%explained earlier, it was known that symmetric WFOMC cannot be in
%PTIME for all queries (by reduction from the spectrum problem), but no
%concrete hardness results were known.  Also, while several
%\#P-hardness results are known for the asymmetric WFOMC, these do not
%extend to the symmetric case, as we argued above.
%

The tractability for \FO{2}~\cite{DBLP:conf/kr/BroeckMD14} raises a
natural question: do other restrictions of FO, like \FO{k} for $k \geq
3$, also have polynomial data complexity?  By carefully analyzing the
details of the construction of $\HardSentence$ we prove that it is
actually in \FO{3}.  This implies a sharp boundary in the \FO{k}
hierarchy where symmetric WFOMC transitions from tractable to
intractable: for $k$ between 2 and~3.
The tractability of $\gamma$-acyclic queries raises another question:
could all conjunctive queries be tractable for symmetric $\WFOMC$?  We answer this also in the negative: we prove that there
exists a conjunctive query $\HardCQ$ for which the symmetric $\WFOMC$
problem is \SHARPONE-hard.  It is interesting to note that the
decision problem associated to $\WFOMC$, namely {\em given $n$, does
  $n \in \texttt{Spec}(\Phi)$?} is trivial for conjunctive queries,
since every conjunctive query has a model over any domain of size $n
\geq 1$.  Therefore, our \SHARPONE-hardness result for $\HardCQ$ is an
instance where the decision problem is easy while the corresponding
weighted counting problem is hard.  We note that, unlike WFOMC, we do not
know the exact complexity of the unweighted, FOMC problem for
conjunctive queries.

% Fagin~\cite{DBLP:journals/jacm/Fagin83} defines several notions of
% acyclicity for hypergraphs, which correspond naturally to notions of
% acyclicity for conjunctive queries without self-joins.  When the
% relational vocabulary consists of only binary relations, then all
% these notions collapse to the classical notion of acyclic graphs.  It
% was known in the case of asymmetric WFOMC that, even for a very simple
% acyclic query, namely $\exists x \exists y \, R(x) \wedge S(x,y) \wedge
% T(y)$, the data complexity is
% \SHARP-hard~\cite{DBLP:journals/jacm/DalviS12}, which immediately
% implies that none of the notions of acyclic queries imply tractable
% data complexity for asymmetric WFOMC.  We prove here that all
% $\gamma$-acyclic queries have polynomial time data complexity for
% symmetric WFOMC.  This raises a natural question: could it be the case
% that all conjunctive queries have polynomial time data complexity in
% the symmetric case?  We answer this in the negative, by proving that
% there exists a conjunctive query whose data complexity is
% \SHARPONE-hard.  

\paragraph*{0-1 Laws}
Our data complexity hardness result sheds some interesting light on
0-1 laws.  Recall that, if $C$ is a class of finite structures and $P$
is a property over these structures, then $\mu_n(P)$ denotes the
fraction of labeled\footnote{The attribute {\em labeled} means that isomorphic structures are counted as distinct; 0-1 laws
  for unlabeled structures also exist. In this paper, we discuss labeled structures only.}  structures in $C$ over a domain
of size $n$ that satisfy the property
$P$~\cite{DBLP:conf/mfcs/KolaitisV00}.  A logic has a 0-1 law over the
class of structures $C$, if for any property $P$ expressible in that
logic, $\mu(P) \eqdef \lim_{n \rightarrow \infty} \mu_n(P)$ is either
0 or 1.  Fagin~\cite{DBLP:journals/jsyml/Fagin76} proved a 0-1 law for
First-Order logic and all structures, by using an elegant transfer
theorem: there exists a unique, countable structure $\mathbf{R}$,
which is characterized by an infinite set of {\em extension axioms},
$\tau$.  He proved that, for every extension axiom, $\lim
\mu_n(\tau)=1$, and this implies $\lim \mu_n(\Phi) = 1$ if $\Phi$ is
true in $\mathbf{R}$, and $\lim \mu_n(\Phi)=0$ if $\Phi$ is false in
$\mathbf{R}$.  Compton~\cite{DBLP:journals/iandc/Compton88} proved 0-1
laws for several classes of structures $C$.  A natural question to ask
is the following: does there exists an elementary proof of the 0-1
laws, by computing a closed formula $\FOMC(\Phi, n)$ for every $\Phi$,
then using elementary calculus to prove that that $\mu_n(\Phi)$
converges to 0 or 1?  For example, if $\Phi = \forall x\exists y
R(x,y)$, then $\FOMC(\Phi, n) = (2^n-1)^n$ and $\mu_n(\Phi) =
(2^n-1)^n/2^{n^2} \rightarrow 0$; can we repeat this argument for
every $\Phi$? On a historical note, Fagin confirms in personal communication that he originally tried to prove the 0-1 law by trying to find such a closed formula, which failed as an approach. 
Our \SHARPONE-result for FO proves that no such
elementary proof is possible, because no closed formula for
$\FOMC(\Phi, n)$ can be computed in general (unless \SHARPONE\ is in
PTIME).

\subsubsection*{Results on the Combined Complexity} Our main result on the
combined complexity is the following.  We show that, for any $k \geq
2$, the combined complexity of FOMC for \FO{k} is \SHARP-complete;
membership is a standard application of Scott's reduction, while
hardness is by reduction from the model counting problem for Boolean
formulas.  Recall that the vocabulary $\sigma$ is always assumed to be
fixed: if it were allowed to be part of the input, then every Boolean
formula is a special case of an \FO{0} formula, by creating a new
relational symbol of arity zero for each Boolean variable, and all
hardness results for Boolean formulas carry over immediately to~\FO{0}.

\subsubsection*{The Associated Decision Problem} We also discuss and
present some new results on the decision problem associated with
(W)FOMC: ``{\em given $\Phi$, $n$, does $\Phi$ have a model over a
  domain of size $n$}?''.  The data complexity variant is, of course,
the spectrum membership problem, which has been completely solved by
Jones and Selman~\cite{jones1972turing}, by proving that the class of
spectra coincides with $\texttt{NETIME}$, that is,
$\setof{\texttt{Spec}(\Phi)}{\Phi \in \mathit{FO}} = \texttt{NETIME}$.  Their
result assumes that the input $n$ is represented in binary, thus the
input size is $\log n$.  In this paper we are interested in the unary
representation of $n$, as $1^n$, which is also called the {\em tally
  notation}, in which case case $\texttt{NETIME}$ naturally identifies
with \NPONE{}.  Fagin proved that, in the tally notation,
$\setof{\texttt{Spec}(\Phi)}{\Phi \in \mathit{FO}} = \NPONE$~\cite[Theorem~6,
Part~2]{fagin1974generalized}.

% ; we note that
% one of our complexity results is $\setof{\FOMC(\Phi)}{\Phi \in FO} =
% \SHARPONE$.

For the decision problem, our result is for the combined complexity:
given both $\Phi, n$, does $n \in \texttt{Spec}(\Phi)$?  We prove that
this problem is NP-complete for \FO{2}, and PSPACE-complete for FO.
The first of these results has an interesting connection to the finite
satisfiability problem for \FO{2}, which we discuss here.  Recall the
classical satisfiability problem in finite model theory: ``{\em given
  a formula $\Phi$ does it have a finite model}?'', which is
equivalent to checking $\texttt{Spec}(\Phi) \neq \emptyset$.
% The second connection is related to the finite satisfiability for
% \FO{2}.  This class of formulas enjoys two nice properties: it has the
% finite model property (it is also called {\em finitely controllable}),
% namely if a formula $\Phi$ in \FO{2} has a model, then it has a finite
% model, and the satisfiability problem is decidable. 
Gr{\"a}del, Kolaitis and Vardi~\cite{DBLP:journals/bsl/GradelKV97}
have proven the following two results for \FO{2}: if a formula $\Phi$
is satisfiable then it has a finite model of size at most exponential
in the size of the sentence $\Phi$, and deciding whether $\Phi$ is
satisfiable is NEXPTIME-complete in the size of $\Phi$.  These two
results already prove that the combined complexity for deciding $n \in
\texttt{Spec}(\Phi)$ cannot be in polynomial time: otherwise, we could
check satisfiability in EXPTIME by iterating $n$ from $1$ to
exponential in the size of $\Phi$, and checking $n \in
\texttt{Spec}(\Phi)$.  Our result settles the combined complexity,
proving that it is NP-complete.

The paper is organized as follows: we introduce the basic definitions
in \autoref{sec:background}, present our results for the data
complexity of the FOMC and WFOMC problems in \autoref{sec:hardness},
present all results on the combined complexity in
\autoref{sec:combined-complexity}, then conclude in
\autoref{sec:conclusions}.

%% file: background.tex
\section{Background}

\label{sec:background}

We review here briefly the main concepts, some already introduced in
\autoref{sec:intro}.

\paragraph*{Weighted Model Counting (WMC)}
%\\[-1ex]

\begin{sloppypar}
The {\em Model Counting} problem is: given a Boolean formula $F$,
compute the number of satisfying assignments $\#F$.  In {\em Weighted
  Model Counting} we are given two real functions $w, \bar w :
\text{Vars}(F) \rightarrow \mathbf{R}$ associating two weights $w(X),
\bar w(X)$ to each variable in $\text{Vars}(F) = \set{X_1, \ldots,
  X_n}$.
% : intuitively, $w(X)$ is the weight of setting $X=1$, while
% $\bar w(X)$ is the weight of setting $X=0$.  
%
The weighted model count
$\WMC(F, w, \bar w)$ is defined as:\\[-2ex]
\begin{align}
  \WMC(F,w, \bar w) \eqdef & \sum_{\theta: \theta(F)=1} W(\theta) \label{eq:wmc}
\end{align}
where, $\forall \theta: \text{Vars}(F) \rightarrow \set{0,1}$:
\begin{align}
W(\theta) \eqdef & \prod_{i: \theta(X_i)=0} \bar w(X_i) \times
\prod_{i: \theta(X_i)=1} w(X_i)\label{eq:def:wmc}\\[-5ex]\nonumber
\end{align}
\end{sloppypar}

\begin{sloppypar}
The model count is a special case $\#F = \WMC(F, 1, 1)$.
\end{sloppypar}

\begin{sloppypar}
  The standard definition of WMC in the literature does not mention
  $\bar w$, instead sets $\bar w = 1$; as we will see, our extension
  is non-essential. When $\bar w = 1$, then we simply drop $\bar w$
  from the notation, and write $\WMC(F, w)$ instead of $\WMC(F, w,
  1)$.  In the {\em probability computation problem}, each variable
  $X_i$ is set to true with some known probability $p(X_i) \in [0,
  1]$, and we want to compute $\Pr(F,p) \eqdef \WMC(F, p, 1-p)$, the
  probability that $F$ is true. All these variations are equivalent,
  because
  of the following identities:\\[-3ex]
\begin{align}
  \WMC(F, w, \bar w) = & \WMC(F, w/\bar w, 1) \times \prod_i \bar w(X_i) \label{eq:wbarw}\\
  \WMC(F, w, \bar w) = & \Pr(F, w/(w+\bar w)) \times \prod_i (w(X_i) +
  \bar w(X_i))\nonumber\\[-5ex]\nonumber
\end{align}
Throughout the paper we write $1$ for the constant function with value
1, and $w_1+w_2$, and $w_1/w_2$ for functions $X\mapsto
w_1(X)+w_2(X)$ and $X \mapsto w_1(X)/w_2(X)$ resp.
\end{sloppypar}

\paragraph*{Weighted First-Order Model Counting (WFOMC)}
%\\[-1ex]

\begin{sloppypar}
Consider FO formulas over a fixed relational vocabulary $\sigma =
(R_1, \ldots, R_m)$ and equality $=$.  Given a domain
size $n$, denote $\Tup(n)$ the set of ground tuples (i.e., ground atoms without equality) over the domain,
thus $|\Tup(n)| = \sum_i n^{\text{arity}(R_i)}$.   
The {\em lineage} of
an FO sentence $\Phi$ refers to a Boolean function $F_{\Phi, n}$ over
$\Tup(n)$ (a ground FO sentence), as well as the corresponding Boolean function over propositional variables referring to ground tuples (a propositional sentence).
It is defined inductively by 
$F_{t, n} = t$ for ground tuples $t$,
$F_{\neg \Phi,n} = \neg F_{\Phi,n}$,
$F_{(\Phi_1 \text{ op } \Phi_2),n} = F_{\Phi_1,n} \text{ op }
F_{\Phi_2,n}$ for $\text{op} \in \set{\wedge, \vee}$, 
$F_{a = b,n} = \texttt{false}$, $F_{a = a,n} = \texttt{true}$
and $F_{\exists x \Phi, n} = \bigvee_{a \in [n]} F_{\Phi[a/x], n}$,
$F_{\forall x \Phi, n} = \bigwedge_{a \in [n]} F_{\Phi[a/x], n}$.  For
any fixed sentence $\Phi$, the size of its lineage is polynomial in
$n$.  Given a domain size $n$ and weight functions $w, \bar w :
\Tup(n) \rightarrow \mathbf{R}$, the Weighted First-Order Model Count
of $\Phi$ is $\WFOMC(\Phi, n, w, \bar w) \eqdef \WMC(F_{\Phi, n},
w, \bar w)$.
\end{sloppypar}

\paragraph*{Symmetric WFOMC}
%\\[-1ex]

%\guy{The introduction only uses one weight function; this section uses
%  two. Perhaps we can keep one weight function throughout the paper,
%  and have its domain be literals instead of atoms?} \dan{I'm trying
%  to avoid introducing literals.  Hope it's clearer now}
\begin{sloppypar}
  In the {\em symmetric} WFOMC, the weight of a tuple depends only on
  the relation name and not on the domain constants.  We call a {\em
    weighted vocabulary} a triple $(\sigma, \mathbf{w}, \mathbf{\bar
    w})$ where $\sigma = (R_1, \ldots, R_m)$ is a relational
  vocabulary and $\mathbf{w} = (w_1,\ldots, w_m)$, $\mathbf{\bar w} =
  (\bar w_1, \ldots, \bar w_m)$ represent the weights (real numbers)
  for the relational symbols.  For any domain size $n$, we extend
  these weights to $\Tup(n)$ by setting $w'(R_i(a_1, \ldots, a_k)) =
  w_i$ and $\bar w'(R_i(a_1, \ldots, a_k)) = \bar w_i$, and we define
  $\WFOMC(\Phi, n, \mathbf{w}, \mathbf{\bar w}) \eqdef \WFOMC(\Phi, n,
  w', \bar w')$.  Throughout this paper we assume that WFOMC refers to
  the symmetric variant, unless otherwise stated.
\end{sloppypar}

For a simple illustration, consider the sentence $\varphi = \exists y
S(y)$.  Then $\WFOMC(\varphi, n, w_S, \bar w_S) = (\bar w_S + w_S)^n -
(\bar w_S)^n$, because the sum of the weights of all possible worlds
is $(\bar w_S + w_S)^n$, and we have to subtract the weight of the
world where $S = \emptyset$.  For another example, consider $\Phi =
\forall x \exists y R(x,y)$.  The reader may check that $\WFOMC(\Phi,
n, w_R, \bar w_R) = \left((w_R+\bar w_R)^n - \bar w^n_R \right)^n$.
In particular, over a domain of size $n$, the formula $\Phi$ has $(2^n
- 1)^n$ models (by setting $w_R = \bar w_R = 1$).

\begin{table*}[t]
  \centering
  \renewcommand{\arraystretch}{1.3} % avoids clashing superscripts with hlines
\begin{tabular}{|p{0.11\textwidth}|p{0.30\textwidth}|p{0.49\textwidth}|}\hline
Problem & Weights for $R$, $S$, and $T$ tuples & Solution for $\Phi = \forall x \forall y(R(x) \vee S(x,y) \vee T(y))$ \\\hline\hline
Symmetric FOMC & $w=\bar w = 1$  & $\FOMC(\Phi,n)= \sum_{k,m = 0,n} {n \choose k}{n \choose   m} 2^{n^2-km}$ \\ \hline
Symmetric WFOMC & $w_{R}$, $w_{S}$, $w_{T}$, \newline $\bar w_{R},\bar
w_{S},\bar w_{T}$ &                           
$\WFOMC(\Phi,n,\mathbf{w},\mathbf{\bar w})= \sum_{k,m = 0,n} {n \choose k}{n \choose   m} W_{k,m}$ \newline
where $W_{k,m}=w_{R}^{n-k} \cdot \bar w_{R}^k \cdot  w_{S}^{km}
\cdot\left(w_{S} + \bar w_{S}\right)^{n^2-km} \cdot w_{T}^{n-m} \cdot
\bar w_{T}^m$\\\hline
Asymmetric WFOMC & $w(R(i), w(S(i,j)), w(T(j))$\newline $\bar w(R(i)),\bar w(S(i,j)), \bar w(T(j))$ \newline depend on $i,j$ & \SHARP-hard for $\Phi$~\cite{DBLP:journals/vldb/DalviS07}\\\hline
\end{tabular}
\caption{Three variants of WFOMC, of increasing generality,
  illustrated on the sentence
  $\Phi = \forall x \forall y(R(x) \vee S(x,y) \vee T(y))$.
  This paper discusses the symmetric  cases only.%\protect{[-10ex]}
%    \guy{The solution for Symmetric WFOMC was wrong in an older table (it was normalized to be a probability). Please check if this one is correct.}
%    \dan{correct}
    }
  \label{table:summary}
\end{table*}

\paragraph*{Data Complexity and Combined Complexity}
\begin{sloppypar}
  We consider the weighted vocabulary $(\sigma, \mathbf{w},
  \mathbf{\bar w})$ fixed.  In the {\em data complexity}, we fix
  $\Phi$ and study the complexity of the problem: {\em given $n$,
    compute $\WFOMC(\Phi, n, \mathbf{w}, \mathbf{\bar w})$}.  In the
  combined complexity, we study the complexity of the problem: {\em
    given $\Phi, n$, compute $\WFOMC(\Phi, n, \mathbf{w}, \mathbf{\bar
      w})$}.  All our upper bounds continue to hold if the weights
  $\mathbf{w}, \mathbf{\bar w}$ are part of the input. We also
  consider the data- and combined-complexity of the associated
  decision problem (where we ignore the weights) {\em given $n$, does
    $\Phi$ have a model over a domain of size $n$?}
\end{sloppypar}

% The data complexity differs in the asymmetric and the symmetric case.
% In both cases $\Phi$ is fixed.  In asymmetric WFOMC the input consists
% of the number $n$ and the functions $w, \bar w$ with domain $\Tup(n)$:
% this problem is in \SHARP, the class of problems computed by a
% nondeterministic counting Turing Machine (defined formally in
% \autoref{sec:hardness}).  In the symmetric WFOMC the weights are fixed
% too (part of the vocabulary), and the input consists only of the
% number $n$: this problem is in \SHARPONE, the subset of \SHARP~where
% the input alphabet is unary.

\paragraph*{Weights and Probabilities}
%\\[-1ex]

% The relationship between the weights $w, \bar w$ and the probability
% $p$ of a Boolean variable is given by $p=w/(w + \bar w)$, or, $w =
% p/(1-p)$, $\bar w = 1/(1-p)$.
While in practical applications the weights are positive real numbers,
and the probabilities are numbers in $[0,1]$, in this paper we impose
no restrictions on the values of the weights and probabilities.  The
definition (\ref{eq:wmc}) of $\WMC(F, w)$ applies equally well to
negative weights, and, in fact, to any semiring structure for the
weights~\cite{kimmig2012algebraic}.  There is, in fact, at least one
application of negative
probabilities~\cite{DBLP:journals/pvldb/JhaS12}, namely the particular
reduction from MLNs to WFOMC described in \autoref{ex:1:cont}: a
newly introduced relation has weight $1/(w-1)$, which is negative when
$w < 1$.  Then, the associated
probability $p = w/(1+w)$ belongs to $(-\infty,0) \cup (1,\infty)$.
%
%In fact, it
%would make no sense to restrict the WMC problem to positive weights
%only: any algorithm that can compute $\WMC(F,w, \bar w)$ for positive
%weights can be extended to compute it for negative weights too, since
%this expression is just a multivariate polynomial in the unknowns $w,
%\bar w$. \guy{In fact, all these algorithms simply require a semiring
%  structure for the weights~\cite{kimmig2012algebraic}} On the other
%hand, allowing for negative weights gives us a powerful tool for
%reducing one WMC to another.
%
% We have already seen this in the last translation from MLNs to
% Symmetric WFOMC (which may result in negative weights) and we have
% more applications of this technique later in the paper.
%
%

As a final comment on negative weights, we note that the complexity of
the symmetric WFOMC problem is the same for arbitrary weights as for
positive weights.  Indeed, the expression $\WFOMC(\Phi, n,
\mathbf{w})$ is a multivariate polynomial in $m$ variables $w_1,
\ldots, w_m$, where each variable has degree $n$. The polynomial has
$(n+1)^m = n^{O(1)}$ real coefficients.  Given access to an oracle
computing this polynomial for arbitrary positive values for $w$, we
can compute in polynomial time all $n^{O(1)}$ coefficients with as
many calls to the oracle; once we know the coefficients we can compute
the polynomial at any values $w_1, \ldots, w_m$, positive or negative.
% 
% 
% 
% for complexity results, it makes no sense to restrict
% the weights to be positive.  More precisely, if the weights are part
% of the input, then the complexity of the $\WFOMC$ problem when the
% weights are unrestricted is polynomial time reducible to the problem
% when the weights are positive.
% 
% 
% Notice that the symmetric $\WFOMC(\Phi, n, w, \bar w)$ is a
% multivariate polynomial in $2m$ variables 
% $w(R_1), \ldots, w(R_m), \bar w(R_1), \ldots, \bar w(R_m)$.  
% Unless otherwise stated, we will
% assume throughout the paper that the relational vocabulary is fixed,
% and therefore the number of variables is $O(1)$.  Each variable
% $w(R_i)$ and $\bar w(R_i)$ has degree $n^{\text{Arity}(R_i)}$, and
% therefore the polynomial has $(n+1)^{\sum_i \text{Arity}(R_i)}$
% coefficients.  We can compute all these coefficients with as many
% calls to an oracle for $\WFOMC(\Phi, n, w, \bar w)$, by varying the
% values of the weights $w(R_i)$, $\bar w(R_i)$.  
% The symmetric WFOMC problem is reducible in polynomial time to the
% problem of computing $n^{O(1)}$ coefficients, and conversely.

For all upper bounds in this paper we assume that the weights $w, \bar
w$, or probabilities $p$, are given as rational numbers represented as
fractions of two integers of $n$ bits each.  We assume w.l.o.g. that
all fractions have the same denominator: this can be enforced by
replacing the denominators by their least common multiplier, at the
cost of increasing the number of bits of all integers to at most
$n^2$.  It follows that the weight of a world $W(\theta)$
(Eq.(\ref{eq:def:wmc})) and $\WMC(F, w, \bar w)$ can be represented as
ratios of two integers, each with $n^{O(1)}$ bits.

%algorithm for the symmetric WFOMC, by changing the weights $w, \bar
%w$.  In other words, computing Eq.(\ref{eq:wfomc}) is equivalent to
%computing the coefficients of the multivariate polynomial.  These
%coefficients have a simple interpretation: for any natural numbers
%$k_1, \ldots, k_m$, the coefficient of the monomial $w_1^{k_1} \cdots
%w_m^{k_m}$ in Eq.(\ref{eq:wfomc}) is precisely the number of possible
%worls with $|R_1| = k_1, \ldots, |R_m| = k_m$ that satisfy $\Phi$.

\paragraph*{Summary}

\autoref{table:summary} summarizes the taxonomy and illustrates the
various weighted model counting problems considered in this paper.
% \begin{sloppypar}
Throughout the rest of the paper, FOMC and WFOMC refer to the
symmetric variant, unless otherwise mentioned.

%% file: data-complexity.tex
\section{Data Complexity}

\label{sec:hardness}

Recall that the language \FO{k} consists of FO formulas with at most
$k$ distinct logical variables.

\subsection{Lower Bounds}

Our first lower bound is for an \FO{3} sentence:

\begin{theorem} \label{th:sharp1:fo3} There exists an \FO{3}
  sentence, denoted $\HardSentence$, s.t.\ the FOMC problem for
  $\HardSentence$ is \SHARPONE-complete. 
%  (2) $\setof{\FOMC(\Phi)}{\Phi \in FO} = \SHARPONE$.
\end{theorem}

% $\FOMC(\Phi)$ denotes the function $n \mapsto \FOMC(\Phi,n)$. We note
% that (2) is an extension of the classic result by Jones and
% Selman~\cite{jones1972turing}, which, restated for the tally notation
% says $\setof{\texttt{Spec}(\Phi)}{\Phi \in FO} = NP_1$ (see
% \cite[Sec.5]{DBLP:journals/bsl/DurandJMM12}).
% 
% We discuss now item (1).  
Van den Broeck et al.~\cite{DBLP:conf/kr/BroeckMD14} have shown that the Symmetric WFOMC problem for
every \FO{2} formula has polynomial time data complexity (the proof is reviewed in
Appendix~\ref{s:fo2-ptime-proof}); Theorem \ref{th:sharp1:fo3} shows
that, unless \SHARPONE\ is in PTIME, the result cannot extend to \FO{k} for $k > 2$.

Our second lower bound is for a conjunctive query, or, dually, a
positive clause without equality. Recall that a {\em clause} is a
universally quantified disjunction of literals, for example $\forall x
\forall y(R(x) \vee \neg S(x,y))$.  A {\em positive clause} is a
clause where all relational atoms are positive.  A {\em conjunctive
  query} (CQ) is an existentially quantified conjunction of positive
literals, e.g. $\exists x \exists y (R(x) \wedge S(x,y))$.  Positive
clauses without the equality predicate are the duals of CQs, and therefore the WFOMC problem is essentially the same for
positive clauses without equality as for CQs.  Note
that the dual of a clause with the equality predicate is a CQ with $\neq$, \eg the dual of $\forall x \forall y (R(x,y) \vee x
= y)$ is $\exists x \exists y(R(x,y)\wedge x\neq y)$.

\begin{corollary} \label{cor:sharp1:cq} There exists a positive clause
  $\HardClause$ without equality s.t.\ the Symmetric
  WFOMC problem for $\HardClause$ is \SHARPONE-hard. Dually, there
  exists a CQ $\HardCQ$ s.t.\ the Symmetric WFOMC
  problem for $\HardCQ$ is \SHARPONE-hard.
\end{corollary}

Corollary~\ref{cor:sharp1:cq} shows that the tractability result for
$\gamma$-acyclic conjunctive queries (discussed below in
Theorem~\ref{th:gamma:acyclic}) cannot be extended to all CQs.  The proof of the Corollary follows easily from three lemmas,
which are of independent interest, and which we present here; the
proofs of the lemmas are in the appendix.  We say that a vocabulary
$\sigma'$ {\em extends} $\sigma$ if $\sigma \subseteq \sigma'$, and
that a weighted vocabulary $(\sigma', \mathbf{w'}, \mathbf{\bar w'})$
{\em extends} $(\sigma, \mathbf{w}, \mathbf{\bar w})$ if $\sigma
\subseteq \sigma'$ and the tuples $\mathbf{w'}, \mathbf{\bar w'}$
extend~$\mathbf{w}, \mathbf{\bar w}$.

\begin{lemma} \label{lemma:exists} 
\begin{sloppypar}
  Let $(\sigma, \mathbf{w}, \mathbf{\bar w})$ be a weighted vocabulary
  and $\Phi$ an FO sentence over $\sigma$.  There exists an extended
  weighted vocabulary $(\sigma', \mathbf{w'}, \mathbf{\bar w'})$ and
  sentence $\Phi'$ over $\sigma'$, such that $\Phi'$ is in
  prenex-normal form with a quantifier prefix $\forall^*$, and
  $\WFOMC(\Phi,n,\mathbf{w},\mathbf{\bar w}) =
  \WFOMC(\Phi',n,\mathbf{w'},\mathbf{\bar w'})$  for all $n$.
\end{sloppypar}
\end{lemma}

\begin{sloppypar}
  This lemma was proven by \cite{DBLP:conf/kr/BroeckMD14}, and says that all
  existential quantifiers can be eliminated.  The main idea is to
  replace a sentence of the form $\forall \mathbf{x} \, \exists y \,
  \psi(\mathbf{x},y)$ by $\forall \mathbf{x} \, \forall y \,(\neg
  \psi(\mathbf{x}, y) \vee A(\mathbf{x}))$, where $A$ is a new
  relational symbol of arity $|\mathbf{x}|$ and with weights $w_A=1,
  \bar w_A=-1$.  For every value $\mathbf{x}=\mathbf{v}$, in a world
  where $\exists y \, \psi(\mathbf{v},y)$ holds, $A(\mathbf{v})$ holds
  too and the new symbol contributes a factor $+1$ to the weight; in a
  world where $\exists y \, \psi(\mathbf{v},y)$ does not hold, then
  $A(\mathbf{v})$ may be true or false, and the weights of the two
  worlds cancel each other out.
\end{sloppypar}

Note that the lemma tells us nothing about the model count of
$\Phi$ and $\Phi'$, since in $\Phi'$ we are forced to set some
negative weights.  If we had $\FOMC(\Phi,n)=\FOMC(\Phi',n)$, then we
could reduce the satisfiability problem for an arbitrary FO sentence
$\Phi$ to that for a sentence with a $\forall^*$ quantifier prefix,
which is impossible, since the former is undecidable while the latter
is decidable.

The next lemma, also following the proof in \cite{DBLP:conf/kr/BroeckMD14}, says that all
negations can be eliminated.

\begin{lemma} \label{lemma:neg}
\begin{sloppypar}
  Let $(\sigma, \mathbf{w}, \mathbf{\bar w})$ be a weighted vocabulary
  and $\Phi$ a sentence over $\sigma$ in prenex-normal form with
  quantifier prefix $\forall^*$.  Then there exists an extended
  weighted vocabulary $(\sigma', \mathbf{w'}, \mathbf{\bar w'})$ and a
  positive FO sentence $\Phi'$ over $\sigma'$, also in prenex-normal
  form with quantifier prefix $\forall^*$, s.t.\ $\WFOMC(\Phi, n,
  \mathbf{w}, \mathbf{\bar w}) = \WFOMC(\Phi', n, \mathbf{w'},
  \mathbf{\bar w'})$ for all $n$.
\end{sloppypar}
\end{lemma}

The idea is to create two new relational symbols $A,B$ for every
negated subformula $\neg \psi(\mathbf{x})$, replace the formula by
$A(\mathbf{x})$, and add the sentence $\forall \mathbf{x}
(\psi(\mathbf{x}) \vee A(\mathbf{x}))\wedge (A(\mathbf{x}) \vee
B(\mathbf{x})) \wedge (\psi(\mathbf{x}) \vee B(\mathbf{x}))$.  By
setting the weights $w_A=\bar w_A = w_B =1$, $\bar w_B=-1$ we ensure
that, for every constant $\mathbf{x}=\mathbf{v}$, either $\neg
\psi(\mathbf{v})\equiv A(\mathbf{v})$, in which case $B(\mathbf{v})$
is forced to be true and the two new symbols contribute a factor $+1$
to the weight, or $\psi(\mathbf{v}) \equiv A(\mathbf{v}) \equiv
\texttt{true}$, in which case $B(\mathbf{v})$ can be either true or
false, and the weights cancel out.

Finally, we remove the $=$ predicate.

\begin{lemma} \label{lemma:eq} Let $(\sigma, \mathbf{w}, \mathbf{\bar
    w})$ be a weighted vocabulary and $\Phi$ a sentence over $\sigma$.
  Then there exists an extended vocabulary $\sigma'$ and sentence
  $\Phi'$ without the equality predicate $=$, such that, for all $n$,
  $\WFOMC(\Phi, n, \mathbf{w}, \mathbf{\bar w})$ can be computed in
  polynomial time using $n+1$ calls to an oracle for $\WFOMC(\Phi', n,
  \mathbf{w'}, \mathbf{\bar w'})$, where $(\sigma', \mathbf{w'},
  \mathbf{\bar w'})$ is an extension of $(\sigma, \mathbf{w},
  \mathbf{\bar w})$.
\end{lemma}

\begin{sloppypar}
  The idea is to introduce a new relational symbol $E$, replace every
  atom $x=y$ with $E(x,y)$, and add the sentence $\forall x\ E(x,x)$.
  Let $\mathbf{w'}, \mathbf{\bar w'}$ be the extension of $\mathbf{w},
  \mathbf{\bar w}$ with $w'_E=z$, $\bar w'_E=1$.  Then $\WFOMC(\Phi',
  n, \mathbf{w'}, \mathbf{\bar w'})$ is a polynomial of degree $n^2$
  in $z$ where each monomial has degree $\geq n$ in $z$, because the
  hard constraint $\forall x\ E(x,x)$ forces $|E| \geq n$.  Moreover,
  the coefficient of $z^n$ is precisely
  $\WFOMC(\Phi,n,\mathbf{w},\mathbf{\bar w})$, because that
  corresponds to the worlds where $|E|=n$, hence it coincides with
  $=$.  We compute this coefficient using $n+1$ calls to an oracle for
  $\WFOMC(\Phi', n, \mathbf{w'}, \mathbf{\bar w'})$.
\end{sloppypar}

Now we give the proof of Corollary \ref{cor:sharp1:cq}.  Starting with
the \SHARPONE-complete sentence $\HardSentence$, we apply the three
lemmas and obtain a positive sentence $\Phi$, with quantifier prefix
$\forall^*$ and without the equality predicate, that is
\SHARPONE-hard.  We write it as a conjunction of clauses, $\Phi=C_1
\wedge C_2 \wedge \cdots \wedge C_k$ (recall that a clause is
universally quantified), and then apply the inclusion-exclusion
formula: $\Pr(\Phi) = \sum_{s \subseteq [k], s\neq
  \emptyset}(-1)^{|s|+1}\Pr(\bigvee_{i\in s} C_i)$.  Since any
disjunction of clauses is equivalent to a single clause, we have
reduced the computation problem $\Pr(\Phi)$ to computing the
probabilities of $2^k-1$ clauses. By duality, this reduces to
computing the probabilities of $2^k-1$ conjunctive queries, $\Pr(Q_1),
\Pr(Q_2), \ldots, \Pr(Q_{2^k-1})$.  We can reduce this problem to that
of computing the probability of a single conjunctive query $\HardCQ$,
by the following argument.  Create $2^k-1$ copies of the relational
symbols in the FO vocabulary, and take the conjunction of all queries,
where each query uses a fresh copy of the vocabulary.  Then
$\Pr(Q_1\wedge \cdots \land Q_{2^k-1}) = \Pr(Q_1) \cdots \Pr(Q_{2^k-1})$,
because now every two distinct queries $Q_i, Q_j$ have distinct
relational symbols.  Using an oracle to compute the probability of
$\HardCQ \eqdef \bigwedge_i Q_i$, we can compute any $\Pr(Q_i)$ by
setting to 1 the probabilities of all relations occurring in $Q_j$,
for $j\neq i$: in other words, the only possible world for a relation
$R$ in $Q_j$ is one where $R$ is the cartesian product of the domain;
assuming $n \geq 1$, $Q_j$ is true, $\Pr(Q_j) = 1$, and hence
$\Pr(\HardCQ) = \Pr(Q_i)$.  We repeat this for every $i$ and compute
$\Pr(Q_1), \ldots, \Pr(Q_{2^k-1})$.  This proves that the CQ $\HardCQ$ is \SHARPONE-hard.  Its dual, $\HardClause$, is a
\SHARPONE-hard positive clause without equality. This proves
Corollary~\ref{cor:sharp1:cq}.

\subsection{Upper Bounds}

A CQ is {\em without self-joins} if all atoms refer to
distinct relational symbols.  
It is standard to associate a hypergraph with CQs, where the variables are nodes, and the atoms
are hyper-edges.  
We define a $\gamma$-acyclic conjunctive query
to be a conjunctive query w/o self-joins whose associated hypergraph
is $\gamma$-acyclic. 
We prove:
\begin{theorem} \label{th:gamma:acyclic} 
\begin{sloppypar}
  The data complexity of Symmetric WFOMC for $\gamma$-acyclic
  CQs is in PTIME.
\end{sloppypar}
%\guy{Should we not repeat here the \SHARPONE{}-hardness of general CQs, as we did for \FO{3}?}
%\dan{I assume you mean \SHARPONE{}-membership, not hardness.  I'm more
%  inclined to remove it from the previous theorem, and move it to
%  sec. 2, and state once and forever that for any FO, the data
%  complexity is in \SHARPONE.}
\end{theorem}
Fagin's definition of
$\gamma$-acyclic hypergraphs~\cite{DBLP:journals/jacm/Fagin83} is reviewed in the proof of Theorem~\ref{th:gamma:acyclic}.  

An open problem is to characterize the conjunctive queries without
self-joins that are in polynomial time.  While no such query has yet
been proven to be hard ($\HardCQ$ in Corollary~\ref{cor:sharp1:cq} has
self-joins), it is widely believed that, for any $k \geq 3$, the
symmetric WFOMC problem for a {\em typed cycle} of length $k$, $C_k =
\exists x_1 \cdots x_k (R_1(x_1,x_2),\allowbreak
R_2(x_2,x_3),\allowbreak\dots,\allowbreak R_k(x_k,x_1))$, is hard.  
\begin{figure*}[t]
  \centering
  \includegraphics[width=\textwidth]{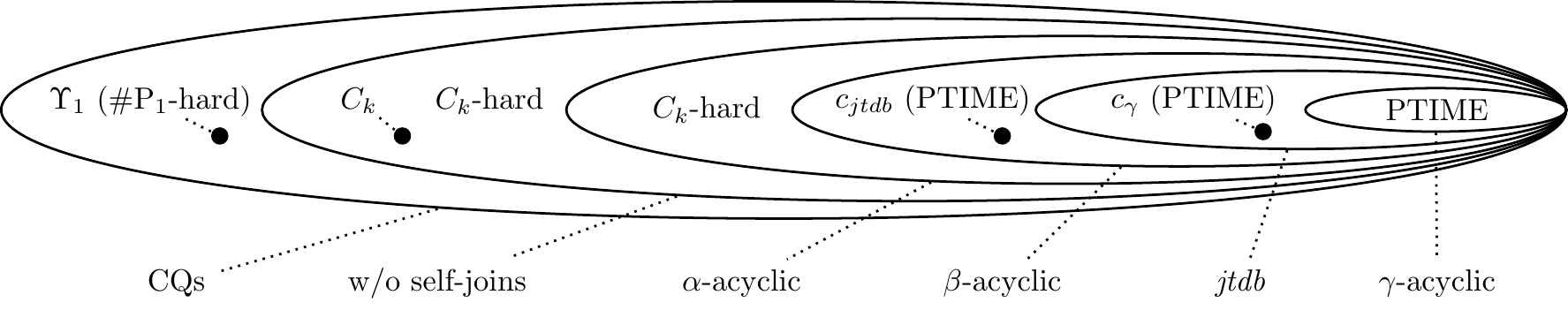}
  \caption{A summary of data complexity results for conjunctive
    queries (or positive clauses). $C_k$-hardness is an informal concept described in the main text. }
  \label{fig:venn}
\end{figure*} 
We
discuss here several insights into finding the tractability border for
conjunctive queries, summarized in
Figure~\ref{fig:venn}.

This boundary does not lie at $\gamma$-acyclicity: the query $c_\gamma
= R(x,z), S(x,y,z), T(y,z)$ is $\gamma$-cyclic (with cycle $RxSyTzR$;
see Fagin~\cite{DBLP:journals/jacm/Fagin83}), yet it still has PTIME
data complexity. The key observation is that $\gamma$-cycles allow the
last variable $z$ to appear in all predicates, turning it into a
\emph{separator variable}~\cite{DBLP:journals/jacm/DalviS12}, hence
$\Pr(Q) = \prod_{a \in [n]} \Pr(Q[a/z])$, which is $[\Pr(Q[a/z])]^n$
by symmetry; $Q[a/z]$ is isomorphic to the query in
Table~\ref{table:summary} and can be computed in polynomial time.  A
weaker notion of acyclicity, called \emph{jtdb} (for join tree with
disjoint branches), can be found in~\cite{duris2012some}. It also does
not characterize the tractability boundary: \emph{jtdb} contains the
$\gamma$-cyclic query above, but it does not contain the PTIME query
$c_{jtdb} = R(x,y,z,u), S(x,y), T(x,z), V(x,u)$.

Fagin~\cite{DBLP:journals/jacm/Fagin83} defines two increasingly
weaker notions of acyclicity: $\beta$- and $\alpha$-acyclic.
$\alpha$-Acyclic queries are as hard as any conjunctive query without
self-joins.  Indeed, if $Q=\exists \mathbf{x} \varphi(\mathbf{x})$ is
a conjunctive query w/o self-joins, then the query $Q'= \exists
\mathbf{x} (A(\mathbf{x})\wedge \varphi(\mathbf{x}))$ is
$\alpha$-acyclic, where $A$ is a new relational symbol, containing all
variables of $Q$.  By setting the probability of $A$ to 1, we have
$\Pr(Q) = \Pr(Q')$. Thus, if all $\alpha$-acyclic queries have PTIME
data complexity, then all conjunctive queries w/o self-joins have
PTIME data complexity.
% Conversely, if there exists a \SHARPONE{}-hard
%conjunctive query w/o self-joins, then there also exists a hard
%$\alpha$-acyclic query.

For all we know, $\beta$-acyclic queries could well coincide with the
class of tractable conjunctive queries w/o self joins.  We present
here some evidence that all $\beta$-cyclic queries are hard, by
reduction from typed cycles, $C_k$.  For that, we need to consider a
slight generalization of WFOMC for conjunctive queries w/o self-joins,
were each existential variable $x_i$ ranges over a distinct domain, of
size $n_i$: the standard semantics corresponds to the special case
where all domains sizes $n_i$ are equal.  We prove that for any
$\beta$-cyclic query $Q$, there exists $k$ such that
$\WFOMC(C_k,\mathbf{n},\mathbf{w},\mathbf{\bar w})$ can be reduced to
$\WFOMC(Q,\mathbf{n'},\mathbf{w'},\mathbf{\bar w'})$. Hence, the
existence of a $\beta$-cyclic query with PTIME data complexity would
imply PTIME data complexity for at least one $C_k$ (informally called
$C_k$-hardness in Fig.~\ref{fig:venn}).
%%%%
The reduction is as follows. By definition, a $\beta$-cyclic query $Q$
contains a weak $\beta$-cycle~\cite{DBLP:journals/jacm/Fagin83} of the
form $R_1 x_1 R_2 x_2 \ldots x_{k-1} R_k x_k R_{k+1}$, where $k \geq
3$, all $x_i$ and $R_i$ are distinct, each $x_i$ occurs in both $R_i$
and $R_{i+1}$, but in no other $R_j$, and $R_{k+1} = R_1$.  Then, we
reduce the WFOMC for $C_k$ to that of $Q$.  First, for each relational
symbol $R_j$ in $Q$, if $R_j$ appears in the cycle then we define
$w_j' = w_j$ and $\bar w_j' = \bar w_j$, otherwise $w'_j = \bar w'_j =
1$.  Second, for all variables $x_i$ that appear in the cycle we set
their domain size $n_i$ to be the same as that of the corresponding
variable in $C_k$, otherwise we set $n_i=1$. Then $Q$ and $C_k$ have
the same WFOMC.
% $\WFOMC(C_k,\mathbf{n},w,\bar w) = \WFOMC(Q,\mathbf{n'},w',\bar w')$.

%
%There is a kind of conjunctive query w/o self-joins, called a typed cycle query $C_k = R_1(x_1,x_2),\allowbreak R_2(x_2,x_3),\allowbreak\dots,\allowbreak R_k(x_k,x_1)$ with $k \geq 3$, that is of particular interest. These queries are considered to be among the hardest queries for which the data complexity is open.
%We can now show that for every $\beta$-cyclic query $Q$, there exists a typed cycle query $C_k$, such that the WFOMC of $C_k$ can be reduced to the WFOMC of $Q$. Any $\beta$-cyclic query with PTIME data complexity would therefore also imply PTIME data complexity for at least one typed cycle query.
%The proof is as follows.
%For a query to be $\beta$-cyclic, it suffices that it contains a \emph{weak $\beta$-cycle}, which is defined as ...
%\guy{TODO: finish}

%
%\dan{we should delete this paragraph}
%Note that the data complexity for both \FO{2} and conjunctive queries differs in the
%symmetric and the asymmetric case.  The WFOMC problem for $\exists x
%\exists y (R(x) \wedge S(x,y) \wedge T(y))$ is \SHARP-hard in the
%asymmetric case~\cite{DBLP:journals/jacm/DalviS12}, but it is in
%polynomial time for the symmetric case (since this query is both in
%\FO{2}, and an acyclic conjunctive query without self-joins).

Finally, we 
%turn our attention to 
discuss a peculiar sentence, whose
complexity we left open in~\cite{GribkofVS:UAI2014}:\\[-3ex]
\begin{theorem} \label{th:c4} 
%\begin{sloppypar}
The data complexity of the symmetric
  WFOMC problem is in PTIME for the query
  \begin{align*}
    Q_{S4} = & \forall x_1 \forall x_2 \forall y_1 \forall
    y_2(S(x_1,y_1) \vee \\ & \neg S(x_2,y_1) \vee S(x_2,y_2) \vee \neg
    S(x_1,y_2))\\[-4ex]
  \end{align*}
\end{theorem}

In~\cite{GribkofVS:UAI2014} we showed that $Q_{S4}$ is in PTIME under
the modified semantics, where $S$ is a bipartite graph.  This implies
that the range of the variables $x_1, x_2$ is disjoint from the range
of the variables $y_1, y_2$.  Now we extended the proof to the
standard semantics used in this paper.  What makes this query
interesting is that the algorithm used to compute it requires a subtle
use of dynamic programming, and none of the existing lifted inference
rules in the literature are sufficient to compute this query.  This
suggests that we do not yet have a candidate for a complete set of
lifted inference rules for the symmetric WFOMC.

\subsection{Proofs}

%In the rest of this section we prove Theorem~\ref{th:sharp1:fo3}.

\paragraph*{Proof of Theorem~\ref{th:sharp1:fo3}}

We briefly recall the basic notions from Valiant's original
papers~\cite{DBLP:journals/tcs/Valiant79,DBLP:journals/siamcomp/Valiant79}.
A {\em counting Turing machine} is a nondeterministic TM with a read-only input
tape and a work tape, that
(magically) prints in binary, on a special output tape, the number of its
accepting computations. The class \SHARPONE{} consists of all
functions computed by some counting TM with polynomial (non-deterministic) running time and a unary
input alphabet. 
A function $f$ is \SHARPONE-hard if, for any function $g$ in
\SHARPONE{} there exists a polynomial time, deterministic TM $\DTM$ with
access to an oracle for $f$ that computes $g$.  Notice that $\DTM$'s
input alphabet is unary.  As usual, $f$ is called \SHARPONE-complete
if it is both hard, and in \SHARPONE.  
% Also, we blur the distinction between TMs and the functions they
% compute.

Our proof of Theorem \ref{th:sharp1:fo3} consists of two steps.  First
we construct a \SHARPONE-complete function $f$, which is computable by
a linear time counting TM $\UTM$, which we call a {\em universal}
\SHARPONE{} machine; in fact, we will define $f$ by describing $\UTM$.
A similar construction in~\cite{DBLP:journals/siamcomp/Valiant79} is
sketched too briefly to see how the particular pairing function can
work; we use a different pairing function and give full details.  To
prove \FO{3} membership, we also need to ensure $\UTM$ runs in
(nondeterministic) linear time, which requires some care given that
the input is given in unary.  Once we have defined $\UTM$, the second
step of the proof is a standard construction of an FO formula to
simulate $\UTM$: we follow
Libkin~\cite[p.~167]{DBLP:books/sp/Libkin04}, but make several changes
to ensure that the formula is in \FO{3}.  The two steps are:

\begin{lemma} \label{lemma:1} There exists a counting TM,
  $\UTM$, with a unary input alphabet, such that (i) $\UTM$ runs in linear
  time, and (ii) the function $f$ that it computes is \SHARPONE-hard.
\end{lemma}

It follows immediately that $f$ is \SHARPONE-complete.

\begin{lemma} \label{lemma:2} 
\begin{sloppypar}
  Let $T$ be any counting TM with a unary input alphabet computing
  some function $f$.  Suppose $T$ runs in time $O(n^a)$. Then there
  exists an \FO{k} formula $\Phi$ over some relational vocabulary
  $\sigma$, s.t.  $f(n) = \FOMC(\Phi, n)/(n!)$, where $k = 3a$ for $a \geq 1$.
\end{sloppypar}
\end{lemma}

Theorem~\ref{th:sharp1:fo3} follows by applying this lemma to $U_1$,
hence $a=1$ and the formula is in \FO{3}.  By allowing
runtimes $O(n^a)$ with $a > 1$, the lemma implies: $\SHARPONE =
\setof{f}{\exists \Phi \in \mathit{FO}, \forall n: f(n) = \lfloor \FOMC(\Phi,
  n)/n!\rfloor}$; this is an extension of the classic result by Jones
and Selman~\cite{jones1972turing}, which, restated for the tally
notation says $\NPONE = \setof{\texttt{Spec}(\Phi)}{\Phi \in \mathit{FO}}$ (see \cite{fagin1974generalized},
\cite[Sec.5]{DBLP:journals/bsl/DurandJMM12}).
By considering $\FOMC$ over unlabeled structures, denoted $\UFOMC$, the correspondence becomes even stronger. In $\UFOMC$, all models that are identical up to a permutation of the constants are counted once, and $\SHARPONE =\setof{\UFOMC(\Phi,n)}{\Phi \in \mathit{FO}}$.

% 
% $\FOMC(\Phi)$ denotes the function $n \mapsto \FOMC(\Phi,n)$. We note
% that (2) is 
% 
% For item (2), the lemma
% implies $\SHARPONE \subseteq \setof{\FOMC(\Phi)}{\Phi \in FO}$, while
% the other containment is straightforward.

% We need a linear runtime for $M$ in order to prove that the formula
% $\Phi$ is in \FO{3}.  To achieve linear runtime, we must allow $\UTM$
% to have multiple tapes: is not possible to simplify $\UTM$ to use a
% single tape and still run in linear time, as we explain below.

\begin{proof} [Proof of Lemma \ref{lemma:1}] The idea for $\UTM$ is
  simple: its input $n$ is represented in unary and encodes two
  numbers $i, j$: $n = e(i,j)$, for some encoding function $e$ to be
  defined below.  $\UTM$ first computes $i, j$ from $n$, then
  simulates the $i$th \SHARPONE{} counting TM on input $j$.  The
  difficult part is to ensure that $\UTM$ runs in linear time: every
  TM $i$ that it simulates runs in time $O(j^{k_i})$ for some exponent
  $k_i$ that depends on $i$, and thus if we construct $\UTM$ naively
  to simply simulate machine $i$ on input $j$, then its runtime is no
  longer polynomial.

  We start by describing an enumeration of counting TMs in \SHARPONE,
  $M_1, M_2, \ldots, M_i, \ldots$, with the property that $M_i$ runs
  in time $\leq (i\cdot j^i + i)^2$ on an input $j$.  
%  In other words,
%  our enumeration ensures that the exponent of $M_i$ is no larger than
%  $2i$.  
  We start by listing all counting TMs over a unary input alphabet in
  standard order $M_1', M_2', \ldots$.  Then we dove-tail pairs of the
  form $M_i=(M'_r,s)$ where $r$ is an index in the standard TM
  order and $s$ is a number.  $M_i$ represents the counting TM that
  simulates $M'_r$ on input $j$ with a timer for $s \cdot j^s + s$
  steps.  The machine $M_i$ can be constructed with at most quadratic
  slowdown over $M'_r$ (due to the need to increment the counter).  We
  further ensure that dovetailing $M_i=(M'_r,s)$ is done such that $i
  \geq s$; for that, it suffices to advance $r$ 
  in such a way that $i$ advances at least as fast as
  $s$, that is, $M_1 = (M'_1,1), M_2 = (M'_2,1), M_3 = (M'_1,2), M_4 =
  (M'_2,2), M_5 = (M'_1,3), \ldots$.  It follows that, for every $i$,
  the runtime of $M_i$ on input $j$ is $\leq (i\cdot j^i + i)^2$.  It
  remains to show that the list $M_1, M_2, \ldots, M_i, \ldots$
  enumerates precisely all \SHARPONE{} functions.  Indeed, each
  function in this list is in \SHARPONE, because the runtime of $M_i$
  is polynomial in the input $j$.  Conversely, every function in
  \SHARPONE{} is computed by some $M_i$ in our list, because it is
  computed by some $M_r'$ whose runtime on input $j$ is $\leq a_r
  \cdot j^{k_r} + b_r$ and this is $\leq s \cdot j^s + s$ if we choose
  $s \eqdef \max(a_r,b_r,k_r)$.  This completes the construction of
  the enumeration $M_1, M_2, \ldots$.

  We describe now the counting machine $\UTM$. Its input is a number $n$
  in unary, which represents an encoding $n = e(i,j)$
  of two integers $i,j$.  We will choose the encoding function $e$ below
  such that it satisfies three properties: (a) $\UTM$ can compute $i,j$
  from $n = e(i,j)$ in linear time (with auxiliary tapes), (b) $e(i,j)
  \geq (i \cdot j^i + i)^2$, and (c) for every fixed $i$, the function
  $j \mapsto e(i,j)$ can be computed in PTIME.  We first prove the
  lemma, assuming that $e$ satisfies these three properties. 

  The counting machine $\UTM$ starts by computing a
  binary representation of its unary input $n$ on its work tape: this step takes
  linear time in $n$.
  Next, it extracts $i, j$ in linear time in $n$ (by property (a)), then it simulates $M_i$
  on input $j$.
  The runtime of the last step is $\leq (i \cdot j^i + i)^2 \leq
  e(i,j)$ (by property (b)), hence $\UTM$ runs in linear time in the
  input $n = e(i,j)$.  It remains to prove that the function $f$
  computed by $\UTM$ is \SHARPONE-hard.  Consider any function $g$ in
  \SHARPONE: we will describe a polynomial-time, deterministic Turing
  machine $\DTM$ with an oracle for $f$ that computes $g$.  Since $g$ is
  in \SHARPONE{} there exists $i$ such that $g$ is computed by $M_i$.  On
  input $j$, $\DTM$ computes $n = e(i,j)$ in PTIME (by property (c)),
  stores it on the oracle tape, then invokes $\UTM$ and obtains the
  result $g(j) = f(n)$.

  It remains to describe the encoding function $e$.  We take $e(i,j) =
  2^{i} 3^{4i \cdot \lceil \log_3 j \rceil} (6j+1)$.  To prove (a), note
  that $i$ is obtained by counting the trailing zeroes in the binary
  representation of $n$, $j$ is obtained by first computing a ternary
  representation of $3^{4i \cdot \lceil \log_3 j \rceil} (6j+1)$,
  ignoring trailing zeros and deriving $j$ from $6j+1$.  (b) $2^{i} 3^{4i
    \cdot \lceil \log_3 j \rceil} (6j+1) \geq  (i \cdot j^i + i)^2$
  follows through direct calculations, using the fact that $3^{4i \cdot
    \lceil \log_3 j \rceil} \geq j^{4i}$.  (c) is straightforward.
\end{proof}

\begin{proof} [Proof of Lemma
  \ref{lemma:2}] \label{s:encodingddiscussion} We describe here the
  most important steps of the proof, and delegate the details to
  Appendix~\ref{app:encoding}.  We will consider only the case $k=1$,
  i.e. the counting TM runs in linear time: the case when $k>1$ is
  handled using a standard technique that encodes $n^k$ time stamps
  using a relation of arity $k$.  We briefly review Trakhtenbrot's
  proof from Libkin~\cite[p.~167]{DBLP:books/sp/Libkin04}: for every
  deterministic TM, there is a procedure that generates a formula $\HardSentence$ such that
  TM has an accepting computation starting with an empty input tape
  iff $\HardSentence$ is satisfiable.
  The signature for
  $\HardSentence$ is (this is a minor variation over Libkin's):
% \\[-5ex]
\begin{align*}
  \sigma = &\set{<, \Min, T_0, T_1, H, (S_q)_{q \in \States(T)}}
% \\[-5ex]
\end{align*}
Then $\HardSentence$ states that (1) $x<y$ is a total order on the domain and
$\Min(x)$ is its
minimum element, (2) $T_0(t,p)$ (or $T_1(t,p)$) is true iff at time $t$ the tape
has a 0 (or a 1) on position $p$, (3) $H(t,p)$ is true iff at time $t$
the head is on position $p$, and $S_q(t)$ is true iff at time $t$ the
machine is in state $q$.  Libkin~\cite{DBLP:books/sp/Libkin04} describes the
sentence $\HardSentence$ that states that all these constraints are
satisfied.

We adapt this to a more general construction that is sufficient to
prove Lemma~\ref{lemma:2}.  We 
address five changes: (1) Our TM is non-deterministic, 
%\dan{this is not a change: Trakthenbrot's encoding also works for
%non-determinstic TMs; plus, people have used this construction in many
%settings, and non-determinism was never an issue. we should remove item (1)}
(2) has $k$ tapes instead of 1, (3) its runtime is
$c\cdot n$ instead of $n$, for some $c>1$, (4) the input tape
initially contains $n$ symbols 1, and (5) $\HardSentence$ needs to be
in \FO{3}.

\begin{sloppypar}
Support for non-deterministic transitions requires only a slight
modification to the sentences.  It is also easy to represent multiple
tapes, by using $k$ different relations $T_{0\tau_i}, T_{1\tau_i}$,
and similarly $k$ head relations $H_{\tau_i}$, for $i=1,k$.  To encode
transitions in \FO{3}, we will assume that the multi-tape TM always
reads or writes only one tape at each time.  This is without loss of
generality: a state that reads and writes all tapes can be converted
into a sequence of $2k$ states that first read one by one each tape
and ``remember'' their symbols, then write one by one each tape and
move their heads.
\end{sloppypar}
% 
% 
% However, we cannot encode a TM transition that reads simultaneously {\em all} tapes, because this cannot be expressed in \FO{3}.  Instead, we will assume w.l.o.g.\ that each state of the TM reads only one single tape: that is, a state $q$ reads only the symbol on the current head position on one single tape, $\tau_q$, ignoring all other tapes.  
% This is without loss of generality, because we can convert any state $q$ that reads all tape symbols in parallel, into a sequence of $2^k$ states that reads the symbols sequentially and ``remembers'' them.  Therefore, for each state $q \in
% \States(T)$, only one tape $\tau_q$ is affected, allowing us to
% encode the transitions in $\HardSentence$ using only 3 variables.
%\guy{I assumed that also only the head on the tape being read in that state
%is moving. Okay? We should mention explicitly.} 
%\dan{I actually assumed that no head is moving during the read phase, but
%instead there is a second write phase when every tape head writes and moves. 
%BUt I wonder how much detail we should give here, it's already very long}

Next, we show how to encode running times (and space) up to $c\cdot n$
for some integer constant $c > 1$, with only a domain of size $n$
available.  The standard way is to increase the arity of the
relations, e.g.\ with arity $a$ we can represent $n^a$ time steps, but
this is not possible within \FO{3}.  Instead, we partition the
computation into $c$ {\em epochs}, each having exactly $n$ time steps,
and similarly we partition the tapes into $c$ regions, each with $n$
cells.  We denote $T_{0\tau er}(t,p)$, $T_{1\tau er}(t,p)$ the
relations $T_0, T_1$ specialized to tape $\tau$, epoch $e$, and region
$r$, and similarly define $H_{\tau er}$ and $S_{qe}$.  Furthermore, we
modify the sentences that encode the TM transition relation to move
the heads across epochs and regions, using only 3 variables.  The
fourth item is easy: we write a formula stating that initially (at
time 1 of epoch 1), region~1 of (input) tape $\tau_1$ is full of 1's,
and all other regions and tapes are full of 0's.  Moreover,
Appendix~\ref{app:encoding} shows that $\HardSentence$ can be written
in \FO{3}.

% Recall that the input domain is $[n]$.  The sentence $\Phi$ states
% that $<$ is a linear order on the input domain $[n]$.
% $\text{\underline{min}}$ is a unary predicate, and $\Phi$ asserts that
% is has a unique value and that this value is smallest element for
% $<$. Since $\UTM$'s runtime is $\leq n^d$, we can encode time stamps as
% $d$-tuples $\mathbf{t} = (t_1, \ldots, t_d)$, and similarly encode
% tape positions by $d$-tuples.  $T_0(\mathbf{p},\mathbf{t})$ is a
% relation of arity $2d$ stating that the tape has the symbol 0 on
% position $\mathbf{p}$ at time $\mathbf{t}$; similarly, $T_1$ records
% the positions having the symbol 1.  $\Phi$ asserts that at time
% $\text{\underline{min}}$ the first $n$ positions are 1, and the rest
% are 0.  For each state $q$ of the TM, the predicate
% $H_q(\mathbf{p},\mathbf{t})$ states that at time $\mathbf{t}$ the
% machine is in state $q$ and the head is on position $\mathbf{p}$.
% Furthermore, $\Phi$ asserts that every transition from $\mathbf{t}$ to
% $\mathbf{t}+1$ is correct according to the transitions of the TM $T$,
% and that the final state is an accepting state,
% see~\cite[p.~167]{DBLP:books/sp/Libkin04} for details.
% 

% Finally, we transform $\sigma$ into a weighted vocabulary by assigning weights $w(R)=\bar w(R)=1$ to every symbol.  
Finally, $\FOMC(\HardSentence, n)$ is precisely the number of accepting computations of the TM on input $n$, times $n!$, coming from the $n!$ ways of ordering the domain.
\end{proof}

\paragraph*{Proof of Theorem \ref{th:gamma:acyclic}}
% Fagin~\cite{DBLP:journals/jacm/Fagin83} defined the notion of a
% $\gamma$-acyclic hypergraph (which we will review below).  Let $Q$ be
% a conjunctive query without self-joins. The hypergraph for $Q$ has
% nodes for the variables of $Q$, and hyperedges for its atoms.  We call
% $Q$ $\gamma$-acyclic if its hypergraph is $\gamma$-acyclic.  
We show how to compute $\Pr(Q)$ rather than $\WFOMC(Q, n)$: we have
seen in Sec.~\ref{sec:background} that these two are equivalent.  We
actually prove the theorem for a more general form of query, where
each variable $x_i$ range over a domain of size $n_i$, thus, $Q =
\exists x_1 \in [n_1], \ldots, \exists x_m \in [n_m] \varphi$, where
$\varphi$ is quantifier-free.  The probability of a query under the
standard semantics (when all variables range over the same domain
$[n]$) is obtained by simply setting $n_1=\cdots=n_m=n$.

%Dually, if $C$ is a positive
%clause, we call it $\gamma$-acyclic if the hypergraph whose nodes are
%the variables, and the edges are the atoms in $C$, is $\gamma$-acylic.
%We prove:
%
%\begin{theorem} For every $\gamma$-acyclic conjunctive query $Q$,
%  $\WFOMC(Q, n, w, \bar w)$ can be computed in polynomial time.  Dually,
%  for every $\gamma$-acyclic positive clause $C$, $\WFOMC(C, n, w, \bar
%  w)$ can be computed in polynomial time.
%\end{theorem}
%
% \begin{proof}

To prove the theorem, we use an equivalent definition of
$\gamma$-acyclicity given by Fagin~\cite{DBLP:journals/jacm/Fagin83},
which we give here together with our algorithm for computing $\Pr(Q)$.
The graph is $\gamma$-acyclic if it can be reduced to an empty graph
by applying the following rules, in any order.
% We show below how to compute $\Pr(Q)$ inductively following these
% rules.
  \begin{enumerate}[(a)]
  \item \label{item:a} If a node $x$ is isolated (i.e., it belongs
    to precisely one edge, say $R(x,y,z)$), then delete $x$. In this
    case we replace the relation $R(x,y,z)$ by a new relation
    $R'(y,z)$, where each tuple has probability $1-(1-p)^{n_x}$, where $p$
    is the probability of tuples in $R$.
  \item \label{item:b} If an edge $R(x)$ is a singleton (i.e., if it
    contains exactly one node), then delete that edge (but do not
    delete the node from other edges that might contain it).  Here, we
    condition on the size $k = |R|$.  For each $k$, let $p_k$ be the
    probability of the {\em residual query} obtained by removing
    $R(x)$ and restricting the range of $x$ to $[k]$.  By symmetry,
    this probability depends only on $k = |R|$, and does not depend on
    the choice of the $k$ elements in the domain.  Then $\Pr(Q) =
    \sum_k {n_x \choose k} p_R^k (1-p_R)^{n_x-k} p_k$, where $p_R$
    denotes the probability of a tuple $\Pr(R(i))$, and is the same
    for all constants $i$ (by symmetry).
  \item \label{item:c} If an edge is empty, $R()$, then delete it.
    We multiply the probability of the residual query by
    $p_R$.
  \item \label{item:d} If two edges (say $R(x,y,z)$, $S(x,y,z)$)
    contain precisely the same nodes, then delete one of these edges.
    Here we replace the two atoms by a new atom $R'(x,y,z)$ whose
    probability is $p_R \cdot p_S$.
\item \label{item:e} If two nodes $x, y$ are edge-equivalent, then
  delete one of them from every edge that contains it. (Recall that
  two nodes are edge-equivalent if they are in precisely the same
  edges.)  Here we replace the two variables $x, y$ by a new variable
  $z$, whose range has size $n_z \eqdef n_x \cdot n_y$.
\end{enumerate}
% \guy{Should we use weights here instead of probabilities?} \dan{I
%   prefer probabilities: items (a) and (d) become harder to understand
%   with weights}

Each operation above is in polynomial time in the size of the binary
representation of the inputs, and there are only polynomially many
operations.  Therefore the entire computation is in polynomial time,
because each intermediate result can be represented using polynomially
many bits.  This follows from the fact that the number of models is
$2^{O(n^a)}$, where $a$ is the maximum arity of any relation in $Q$,
hence the number of models can be represented using $O(n^a) =
n^{O(1)}$ bits.

\begin{example} Consider the following linear chain query:
  \begin{align*}
    Q = & \exists x_0 \exists x_1 \cdots \exists x_m
    R_1(x_0,x_1)\wedge \cdots R_m(x_{m-1},x_m)
  \end{align*}
  where the probabilities of the $m$ relations are $p_1, \ldots, p_m$.
  Denote $P_{n_0, \ldots, n_m}$ the probability of $Q$ when the
  domains of $x_0, x_1, \ldots, x_m$ are sets of sizes $n_0, n_1,
  \ldots, n_m$ (thus, initially $n_0 = n_1 = \cdots = n_m = n$).  Then
  the variable $x_m$ is isolated (item \ref{item:a}), hence we can
  eliminate it and update the probability of $R_m$ to
  $1-(1-p_m)^{n_m}$.  Now $R_m$ is a singleton relation, hence we can
  remove it (item \ref{item:b}), and restrict the domain of $x_{m-1}$
  to have size $k_{m-1}$, for $k_{m-1} = 1,n_{m-1}$.  Therefore:
  \begin{align*}
    P_{n_0, \ldots, n_{m-1}, n_m} =   \sum_{k_{m-1}=1,n_{m-1}} & P_{n_0,  \ldots, n_{m-2},k_{m-1}} \cdot{n_m \choose k_m} \\
   & \cdot [1-(1-p_m)^{k_m}]^{k_{m-1}} \\ & \cdot [(1-p_m)^{k_m}]^{n_{m-1}-k_{m-1}}
  \end{align*}
  Repeating this process we arrive at an expression that is computable
  in polynomial time in $n$ (for a fixed $m$).  Notice that this
  formula does not appear to be computable in polynomial time in both
  $n$ and $m$.  We leave open the combined complexity of acyclic
  queries.
\end{example}
%
%Rule \ref{item:b} is the atom-counting rule, while \ref{item:c} is
%trivial.  The other rules seem to require new probabilistic rules,
%albeit simple ones.  For example, for \ref{item:a}, if $x$ occurs only
%in the relation $S(x,\ldots)$, then replace $p(S)$ by $1-(1-p(S))^n$,
%while \ref{item:e} requires us to replace two variables $x, y$ by a
%new variable $z$, with a new domain, of size $n^2$.

\paragraph*{Proof of Theorem \ref{th:c4}}

%Finally, our last result is to prove that a certain query is
%tractable, although, to the best of our knowledge, no rules exists
%to compute it.  Instead, tractability follows from a dynamic
%programming algorithm on the domain $[n]$:
%
%\begin{theorem}
%  Let $C_4 = \forall x_1, x_2, y_1, y_2(S(x_1,y_1) \vee \neg
%  S(x_2,y_1) \vee S(x_2,y_2) \vee \neg S(x_1,y_2))$.  Then $\WFOMC(C_4,
%  n, w, \bar w)$ can be computed in polynomial time.
%\end{theorem}
%

First note that, by using resolution, the query implies the following
statement, for every $k\geq 2$:
\begin{align}
  \forall x_1, y_1, \ldots, x_k, y_k (&S(x_1,y_1) \vee \neg S(x_2,y_1) \nonumber \\
\vee &S(x_2,y_2) \vee \neg S(x_2,y_3) \nonumber\\
\vee &\ldots \vee \neg  S(x_1,y_k)) \label{eq:sk}
\end{align}

For any two numbers $n_1, n_2$, denote $Q_{n_1n_2} = \forall x_1 \in
[n_1], \forall x_2 \in [n_1], \forall y_1 \in [n_2], \forall y_2 \in
[n_2],$ $(S(x_1,y_1) \vee \neg S(x_2,y_1) \vee S(x_2,y_2) \vee \neg
S(x_1,y_2))$, in other words we restrict the range of the variables to
some domains $[n_1], [n_2]$.  These domains are not required to be
disjoint, instead we use the standard assumption $n_1 \leq n_2$
implies $[n_1] \subseteq [n_2]$.  When $n_1=n_2=n$ then $Q_{n_1n_2}$
is equivalent to $Q_{S4}$.
%
%For simplicity of exposition, below we drop the set
%membership conditions on $x,y$ and implicitly quantify over $X,Y$.
%  %we assume below that all variables $x$ are in $X$ and $y$ are in $Y$.
%
We claim the following.  If $D$ is a model of $Q_{n_1n_2}$, then
either property $P_a$ or $P_b$ holds in~$D$:
\begin{align*}
  P_a \equiv & \exists x \in [n_1], \forall y \in [n_2], S(x,y) \\
  P_b \equiv & \exists y \in [n_2], \forall x \in [n_1], \lnot S(x,y)
\end{align*}
Suppose not. Consider any model of $Q_{n_1n_2}$ that does not satisfy
$P_a, P_b$.
%Let $n_1=|X|$ and $n_2=|Y|$, and assume
%without loss of generality that $n_1 \le n_2$.  
Pick any element $x_1 \in [n_1]$. As $P_a$ does not hold, $\exists y_1
\in [n_2] \lnot S(x_1,y_1)$.  As $P_b$ does not hold, $\exists x_2 \in
[n_1], S(x_2, y_1)$.  Continuing, $\exists y_2 \in [n_2]$, $\lnot S(x_2,
y_2)$, $\exists x_3 \in [n_1], S(x_3, y_2)$ and $\exists y_3 \in
[n_2]$, $\lnot S(x_3,y_3)$.
%We
%also have $x_3 \ne x_1$, or else $x_1,x_2,y_1,y_2$ falsify $Q_{S4}$,
%and likewise, $y_3 \ne y_1$.
  %Let $n_x$ be the number of distinct $x$ elements and $n_y$ the number 
  %of distinct $y$ elements. In the symmetric setting,
  %$n_x=n_y=n$, but the recurrence derived here for $S_4$ holds when
  %$x,y$ may be quantified over differing domains.
Continuing, we obtain an arbitrarily long sequence of values $x_1,
y_1, x_2, \ldots$ such that: $\lnot S(x_1,y_1)$, $S(x_2,y_1)$, $\lnot
S(x_2,y_2)$, $\dots$, $S(x_{n_1}, y_{n_1-1})$, $\lnot S(x_{n_1},
y_{n_1})$.  Note that we can never have $x_i=x_j$ or $y_i=y_j$, for
$i\neq j$, because that would violate Eq.(\ref{eq:sk}) for $k=j-i$.
Since the domain is finite, this is a contradiction.

Therefore, either $P_a$ or $P_b$ holds.  Clearly, both statements
cannot hold, as they are exclusive events.  Denote $f$ and
$g$ the weighted model count for $Q_{n_1n_2}$ in these two cases:
\begin{align*}
  f(n_1,n_2) = & \WFOMC(Q_{n_1n_2} \wedge P_a, n, w, \bar w) \\
  g(n_1,n_2) = & \WFOMC(Q_{n_1n_2} \wedge P_b, n, w, \bar w)
\end{align*}
%
%
%only one of these
%two statements can hold on a given structure $D$.  Let $f(n_1, n_2)$
%
%: since (b) does not hold, we must have
%$\exists x^\prime S(x^\prime, y_n)$, and thus $x_n, x^\prime, y_{n-1},
%y_n$ falsify $Q$.
%
%The above property continues to hold on substructures of $Q$ induced
%by deleting an element $x$ satisfying condition (a) or by deleting an
%element $y$ satisfying condition (b). This justifies the following
%recurrence:
Then we have $\WFOMC(Q_{n_1n_2},n, w, \bar w) = f(n_1,n_2) +
g(n_1,n_2)$.  It remains to show how to compute $f, g$.

Consider a model that satisfies $P_a$, hence the set $X =
\setof{x}{\forall y \in [n_2], S(x,y)}$ is non-empty, hence $k=|X|
\geq 1$.  Remove the elements $X$ from the domain $[n_1]$ (and rename
the elements such that $[n_1]-X= [n_1-k]$) and call $D'$ the resulting
substructure.  Then $D'$ still satisfies the query $Q_{(n_1-k),n_2}$,
and, by the removal of all elements $X$, cannot satisfy $P_a$, hence it must satisfy $P_b$.
This justifies the following recurrence, completing the proof of
Theorem \ref{th:c4}:\\[-3ex]
  \begin{align*}
    f(n_1, 0) &= 1 &
    f(n_1, n_2) &= \sum_{k=1}^{n_1} {n_1 \choose k} w^{kn_2}  g(n_1-k,n_2) \\
    g(0, n_2) &= 1 &
    g(n_1, n_2) &= \sum_{\ell=1}^{n_2} {n_2 \choose \ell} {\bar w}^{n_1\ell}  f(n_1,n_2-\ell)\\[-4ex]
  \end{align*}

%%  some queries that we should be able to compute:
%%  
%%  \begin{align*}
%%  q_1 =  & A(x,y,z,u) \vee B(x,y) \vee C(z,u) \vee D(x) \vee E(y) \vee F(z)  \vee H(u) \\
%%  q_2 =  & R(x,y) \vee S(y,z) \vee T(z,u) \vee K(u,w) \\
%%  q_3 =  & R(x,y,y') \vee S(y,y',z,z',z'') \vee T(z, z', z'', u) \vee K(u,w) \\
%%  \end{align*}
%%  
%%  my intuition is that the algorithm is most naturally expressed when
%%  queries are ``typed'', i.e. different variables may have different
%%  types, in other words range over different domains.  For example, in
%%  $q_3$ above we can replace $y,y'$ with a new variable $v$ that ranges
%%  over $[n^2]$, and similarly replace $z,z',z''$ with a new variable
%%  that ranges over $[n^3]$: the rewritten query looks like $q_2$, but
%%  now variables have different types.
%%  

%% file: combined-complexity.tex
\section{Combined Complexity}

\label{sec:combined-complexity}

In the combined complexity we consider a fixed vocabulary $\sigma =
(R_1, \ldots, R_m)$, and assume that both $\Phi$ and $n$ are given as
part of the input. As before, $n$ is given in unary (tally) notation.
We consider both the FOMC problem, ``{\em compute
  $\FOMC(\Phi, n)$}'', and the associated decision problem ``{\em is
  $n \in \texttt{Spec}(\Phi)$?}''.  Our upper bound for FOMC also
holds for WFOMC.  Recall that the spectrum $\texttt{Spec}(\Phi)$ of a
formula $\Phi$ is the set of numbers $n$ for which $\Phi$ has a model
over a domain of size~$n$.

\begin{sloppypar}
Vardi~\cite{DBLP:conf/stoc/Vardi82} proved that the model checking
problem, ``{\em given $\Phi$ and a structure $D$, is $\Phi$ true in
  $D$?}'' is PSPACE-complete.  This implies that the above decision
problem is also in PSPACE: to check $n \in \texttt{Spec}(\Phi)$
enumerate over all structures $D$ of a domain of size $n$, and check
if $\Phi$ is true in $D$.  By the same argument, FOMC is also in
PSPACE.  We prove:
\end{sloppypar}

% 
% 
% A natural goal is to find the exact border between tractability and
% intractability for symmetric WFOMC.  This border is known precisely
% for asymmetric WFOMC in the case of UCQ queries (and their duals:
% Positive, $\forall^*$ formulas)~\cite{DBLP:journals/jacm/DalviS12},
% but in the symmetric case it is largely open.  In this section we
% present three results that help delineate that border.  For the
% combined and query complexity results in this section we assume that
% the vocabulary is part of the input. However, we do bound the arity of
% the vocabulary, to ensure then the total number of models is
% representable using polynomially many bits.  Our first result is:

\begin{theorem} \label{th:fo2} 
\begin{sloppypar}
  (1) For every $k\geq 2$, the combined
  complexity for FOMC for \FO{k} is \SHARP-complete.
  % (2) The data complexity for symmetric WFOMC for \FO{k} is in
  % \SHARPONE{}, and there exists a sentence in \FO{3} whose data
  % complexity is \SHARPONE{} hard.
(2) The combined complexity for the decision problem $n \in
\texttt{Spec}(\Phi)$ is NP-complete for \FO{2}, and is PSPACE-complete
for FO.
\end{sloppypar}
%  (2) ~\cite{broeck2011completeness,DBLP:conf/kr/BroeckMD14} The data
%  complexity for symmetric WFOMC for \FO{2} is in PTIME.
\end{theorem}

\begin{sloppypar}
The \SHARP-membership in (1) also holds for the WFOMC problem.  Recall
that the vocabulary $\sigma$ is fixed.  If $\sigma$ were allowed to be
part of the input, then the lower bound in (1) follows immediately
from the \SHARP-hardness result for \#SAT, because any Boolean formula
is trivially encoded as an \FO{0} formula, by introducing a new,
zero-ary relational symbol for every Boolean variable.
\end{sloppypar}

%In the rest of the section we prove the three theorems.

\paragraph*{Proof of Theorem \ref{th:fo2}}

%For item (2), membership in \SHARPONE{} is immediate, while hardness was
%proven in~\autoref{th:sharp1:fo3}.
We start by proving item (1) of Theorem \ref{th:fo2}.
% 
% Hardness follows by reduction
% from the model counting problem for positive, partitioned, 2CNF
% propositional formulas, i.e. formulas of the form
% $\bigwedge_{(i,j)\in E} (X_i \vee Y_j)$, which is equivalent to the
% \FO{1} formula $\bigwedge_{(i,j)\in E} \forall x,y (R_i(x) \vee
% S_j(y))$ over the unary alphabet $R_1, R_2, \ldots, S_1, S_2, \ldots$
% and a domain of size $n=1$.
% First note that there are polynomially many  that do not appear in the lineage of the sentence 
To prove membership in \SHARP, it suffices to show that the lineage of
a sentence $\varphi$ of size $s$ over a domain of size $n$ is
polynomial in $s$ and $n$, then use the fact that WMC for Boolean
functions is in \SHARP.  However, \FO{k} formulas have, in general,
exponentially large lineage, e.g. the formula checking for the
existence of a path of length $n$, $\exists x \exists y(R(x,y) \wedge
\exists x(R(y,x) \wedge \exists y(R(y,x) \wedge \ldots)))$, over a
domain of size $n$ has lineage of size $\Omega(n^n)$.  Instead, we
first transform the formula by removing all nested variables.  For
that, we apply Scott's reduction, which we give below, following the
presentation by Gr{\"a}del, Kolaitis, and
Vardi~\cite[Prop.3.1]{DBLP:journals/bsl/GradelKV97}; while Scott's
reduction was described for \FO{2}, it carries over unchanged to
\FO{k}.  More precisely, the reduction converts a sentence $\varphi$
of size $s$ into a new sentence $\varphi^*$ over an extended
vocabulary, satisfying the following properties:

\begin{enumerate}
\item \label{item:scott:1} \begin{sloppypar}
  The finite models of $\varphi$ and
  $\varphi^*$ are in one-to-one correspondence, and the corresponding
  models have the same weight.
  \end{sloppypar}
\item \label{item:scott:2} $\varphi^*$ has size $O(s)$.
\item \label{item:scott:3} $\varphi^*$ is a conjunction of sentences
  in prenex normal form, i.e. $Q_1 x_1 Q_2 x_2 \cdots Q_k x_k \psi$
  where each $Q_i$ is either $\forall$ or $\exists$, and $\psi$ is
  quantifier-free.
\end{enumerate}

The new formula has a lineage of size $O(n^k s)$, because its
quantifier depth is bounded by $k= O(1)$, which implies WFOMC is in
\SHARP.  It remains to describe Scott's reduction, which we review
here briefly, for completeness.  Introduce a new relational symbol
$S_\psi$ for every subformula $\psi$ of $\varphi$, where the arity of
$S_\psi$ equals the number of free variables in $\psi$, and define the
sentence $\theta_\psi \equiv \forall x_1 \cdots \forall
x_\ell(S_\psi(x_1, \ldots, x_\ell) \Leftrightarrow \theta_\psi')$,
where $\theta'_\psi$ depends on the structure of $\psi$ as follows: if
$\psi$ is an atomic formula, then $\theta'_\psi = \psi$, if
$\psi=\psi_1 \wedge \psi_2$ then $\theta'_\psi = S_{\psi_1} \wedge
S_{\psi_2}$, if $\psi = \neg \psi_1$ then $\theta'_\psi = \neg
S_{\psi_1}$ and if $\psi = \forall x \psi_1$ then $\theta'_\psi =
\forall x S_{\psi_1}$.  The new formula $\varphi^*$ is defined as
$S_\varphi \wedge \bigwedge_\psi \theta_\psi$.  By setting $w(S_\psi)
= \bar w(S_\psi) = 1$ for all new symbols, we ensure that the models
of $\varphi$ and $\varphi^*$ are not just in one-to-one
correspondence; they have the same weights.

% We have so far assumed that all tuples in the vocabulary appear in the lineage of $\varphi$. If this is not the case, note that the contribution of missing tuples to the WFOMC can be computed in PTIME independently  (there are polynomially many since arity is bounded).
% \dan{I propose we remove this, because it is obvious}

Next we prove \SHARP-hardness for \FO{2} (this implies hardness for
\FO{k} for every $k \geq 2$).  We use reduction from \#SAT: given a
Boolean formula $F$ over $n$ variables $X_1, \ldots, X_n$, compute
$\#F$.  This problem is
\SHARP-hard~\cite{DBLP:journals/tcs/Valiant79}.

% For the {\em query complexity}, we simply consider $n$ zero-ary
% predicates $X_1, \ldots, X_n$ and note that the propositional formula
% $F$ is a special case of an \FO{1} formula (in fact, \FO{0}), hence we
% simply set $\varphi=F$, and note that over a domain of size 1,
% $\FOMC(\varphi, 1) =\#F$.  Notice that here the input vocabulary is
% not fixed.

% Next, we prove \SHARP-hardness for the {\em combined complexity} over
% a fixed input vocabulary.  
\begin{sloppypar}
  Define the vocabulary $\sigma$ consisting of 3 unary symbols $A, B,
  C$, and 2 binary symbols $R, S$.  Given a Boolean formula $F$, we
  construct an \FO{2} sentence $\varphi_F$ such that, over a domain of
  size $n+1$, the number of models of $\varphi_F$ is $\FOMC(\varphi_F,
  n+1) = (n+1)!  \cdot \#F$.
\begin{figure}[t]
  \centering
  \includegraphics{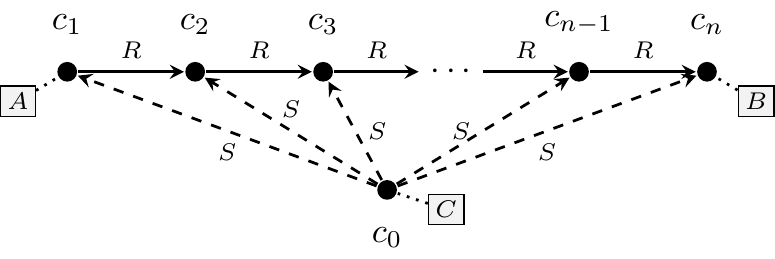}
  \caption{All models of $\varphi_F$ represent graphs of the depicted
  form. There is a linear graph of $R$-edges between the distinguished
  nodes $A$ and $B$. The $S$-edges go from a distinguished node $C$ to
  all others. They are optional, and are in one-to-one correspondence
  with the variables $X_i$ in $F$.
  %\\[-5ex]
  }
  \label{fig:graph}
\end{figure} 
The sentence
$\varphi_F$ enforces a particular graph structure, as illustrated in Figure~\ref{fig:graph}, by asserting the following:
\end{sloppypar}

%% this table should appear *before* the paragraph
%% it is placed here in the input to force its position at the top of the next page.
\begin{table*}
  \centering
  \begin{tabular}{|l|l|} \hline
    Untyped triangles & $\exists x, y, z(R(x,y), R(y,z), R(z,x))$ \\ \hline
    Typed triangles (3-cycle) & $\exists x, y, z(R(x,y), S(y,z), T(z,x))$ \\ \hline
    $k$-cycle, for $k \geq 3$ & $\exists x_1, \ldots, x_k (R_1(x_1,x_2), R_2(x_2,x_3), \ldots,  R_k(x_k,x_1)$ \\\hline
    Transitivity  & $\forall x,y,z (E(x,y) \land E(y,z) \Rightarrow E(x,z))$ \\ \hline
    Homophily  & $\forall x,y,z (R(x,y) \land S(x,z) \Rightarrow R(z,y))$ \\ \hline
    Extension Axiom (Simplified) & 
    $\!\begin{aligned}
      & \forall x_1, x_2, x_3 (x_1\neq x_2 \wedge x_1 \neq x_3 \wedge x_2 \neq x_3 \,\,\Rightarrow\,\, \exists y E(x_1,y) \wedge E(x_2,y) \wedge E(x_3,y))
    \end{aligned}$ \\ \hline
%    Complex Extension Axiom & $\begin{aligned}
%      & x_1 \neq x_2 \wedge \ldots \wedge x_4 \neq x_5 \\
%      & \quad \Rightarrow \exists y E(x_1,y) \wedge E(x_2,y) \wedge
%      \dots \wedge \neg E(x_4,y) \wedge \neg E(x_5,y)
%    \end{aligned}$ \\ 
% \hline
\end{tabular}
\caption{A list of open problems: for each formula it is conjectured
  that FOMC is hard.
}
  \label{tab:open}
\end{table*}

\begin{itemize}
\item There exists three unique, distinct elements $x, y, z$ such that
  $A(x), B(y), C(z)$ are true:
  \newline $\exists x A(x) \wedge \forall x, \forall y (A(x) \wedge
  A(y)) \Rightarrow x=y$, and similarly for $B$ and $C$;
  \newline
  $\neg \exists x (A(x) \wedge B(x))$ and similarly for $A,C$, and
  $B,C$.
\item There exist $n$ elements $x_1, \ldots, x_n$ such that the
  following holds:
  \newline $A(x_1), R(x_1,x_2),R(x_2,x_3), \ldots,
  R(x_{n-1},x_n),B(x_n)$.
  \newline This is expressible in \FO{2}, by reusing variables.
\item For every number $m \in [2n]-\set{n}$, it is not the case that
  there exists $m$ elements $x_1, \ldots, x_m$ such that:
  \newline $A(x_1), R(x_1,x_2),R(x_2,x_3), \ldots, R(x_{m-1},x_m),B(x_m)$.
\item For all $x,y$, if $R(x,y)$ then neither $C(x)$ nor $C(y)$.
\item For all $x,y$, if $S(x,y)$ then $C(x)$.
\item Finally, $\varphi_F$ contains a statement obtained from $F$ by
  replacing each Boolean variable $X_i$ by the sentence $\gamma_i
  \eqdef \exists x, \exists z (S(z,x) \wedge \alpha_i(x))$, where
  $\alpha_i(x)$ is the following formula with free variable $x$: there
  exists a path $A(x_1), R(x_1,x_2),\ldots, R(x_{i-1},x)$ (if $i=1$
  then $\alpha_1(x) \equiv A(x)$).
%   Similarly, denote $\beta_i(x)$ the
%   following formula with free variable $x$: there exists $x_{i+1},
%   \ldots, x_n$ such that $R(x,x_{i+1}), \ldots, R(x_{n-1},x_n),B(x_n)$
%   (if $i=n$ then $\beta_n(x) \equiv B(x)$).  
\end{itemize}

The reader may check that, for any database instance $D$ over a domain
of size $n+1$ that satisfies $\varphi_F$ there exists a unique
permutation $c_0, c_1, \ldots, c_n$ over its domain such that the
relations $A, B, C$ and $R$ contain precisely the following tuples:
$C(c_0), A(c_1), B(c_n)$, $R(c_1,c_2), \ldots,$ $R(c_{n-1},c_n)$.
Indeed, if it contained some $R(c_i,c_j)$ with $j\neq i+1$, then $c_1,
\ldots, c_i, c_j, c_{j+1}, \ldots, c_n$ forms a path from $A$ to $B$
of some length $m \leq 2n$ and $m\neq n$, which contradicts the
sentence $\varphi_F$; notice also that $\alpha_i(x)$ is true iff
$x=c_i$.  Therefore the only relation that is left unspecified in $D$
is $S$, which may contain an arbitrary number of tuples of the form
$S(c_0,c_i)$.  These tuples are in one-to-one correspondence with the
Boolean variables $X_i$, proving our claim.

\begin{sloppypar}
Now we prove item (2) of Theorem \ref{th:fo2}.  The claim for \FO{2}
follows immediately from the proof above.  It remains to prove that
the combined complexity for FO is PSPACE, for which we use a reduction
from the Quantified Boolean Formula (QBF) problem, which is known to
be PSPACE complete.  A {\em Quantified Boolean Formula} is a formula
of the form $Q_1 X_1 Q_2 X_2 \ldots Q_n X_n F$ where each $Q_i$ is a
quantifier $\forall$ or $\exists$, and $F$ is a Boolean formula over
the variables $X_1, \ldots, X_n$.  We make the following change to the
construction above.  Recall that a Boolean variable $X_i$ in $F$ was
represented by $S(c_0,c_i)$.  Now we extend $S$ to a ternary relation
$S(x,y,u)$, restrict $u$ to two constants (we choose $c_1$ and $c_n$
arbitrarily) and represent $X_i$ by $S(c_0,c_i,c_1)$ and $\neg X_i$ by
$S(c_0,c_i,c_n)$.  Then, we replace the quantifiers $\forall X_i$ or
$\exists X_i$ with $\forall u$ or $\exists u$.  More precisely, the
new formula $\varphi_F$ contains the following statements:
\end{sloppypar}

\begin{itemize}
\item \begin{sloppypar}
  If $S(x,y,u)$ is true, then $u$ is either the distinguished $A$
  or the distinguished $B$ element: $\forall x, y, u (S(x,y,u)
  \Rightarrow A(u) \vee B(u))$.
  \end{sloppypar}
\item \begin{sloppypar}
  If $u,v$ are the distinguished $A$ and $B$ elements, then
  $S(x,y,u)$ is the negation of $S(x,y,v)$:
  $\forall u,v,x,y (A(u) \wedge B(v) \Rightarrow (S(x,y,u) \texttt{ xor
} S(x,y,v)))$.
  \end{sloppypar}
\end{itemize}

\begin{sloppypar}
Finally, we rewrite a QBF $\forall X_i(\ldots)$ into $\forall u (A(u)
\vee B(u) \Rightarrow \ldots)$ and a QBF $\exists X_i (\ldots)$ into
$\exists u ((A(u) \vee B(u)) \wedge \ldots)$.  We omit the
straightforward details.
\end{sloppypar}

%% file: conclusions.tex
\section{Conclusions}
\label{sec:conclusions}

In this paper we discuss the symmetric Weighted FO Model Counting
Problem.  Our motivation comes from probabilistic inference in Markov
Logic Networks, with applications to modern, large knowledge bases,
but the problem is also of independent theoretical interest. We
studied both the data complexity, and the combined complexity.  For
the data complexity we established for the first time the existence of
an FO sentence for which the Symmetric Model Counting problem is
\SHARPONE-hard, and also the existence of a Conjunctive Query for
which the Symmetric Weighted Model Counting problem is \SHARPONE-hard.
We also showed that for all $\gamma$-acyclic conjunctive queries WFOMC
can be computed in polynomial time.  For the combined complexity, we
proved a tight bound of \SHARP-completeness for \FO{2}.  We also
discussed the associate decisions problem.

We end this paper with a list of open problems, listed in Table
\ref{tab:open}: for each query in the table, the complexity of the
FOMC or the WFOMC problem is open.

% TABLE IS MOVED EARLIER TO HAVE IT APPEAR ON THE SAME PAGE

\paragraph*{Acknowledgments} We thank Ronald Fagin, Phokion Kolaitis and Lidia Tendera for
discussions on topics related to this paper.  This work was partially
supported by NSF IIS-1115188, IIS-0911036, CCF-1217099, and the Research Foundation-Flanders (FWO-Vlaanderen).

%\guy{Need to say something about Extension Axioms. Why is $y$ not a
%  separator?}  \dan{there are inequality conditions on the left of the
%  implication, so I don't know how to compute it.  If we can compute
%  the extension axioms that would be a huge result, since the
%  extension axioms are all that's needed for all 0-1 laws.}

%\guy{Can we show that FO$^2$+transitivity is hard?}
%
%\guy{People claim with Trakhtenbrot's theorem that you only need a single binary relation. This is no longer true for our \#P$_1$ construction?}
%
%\guy{Finite satisfiability research is often motivated by the fact
%that you can optimize a query if it has unsatisfiable parts. Can a
%similar case be made for FO model counting? Can queries be
%optimized/reordered by knowing their relative model count or `degree
%of satisfiability'?}  \dan{This is a cool idea: one can optimize a
%query based on the fact that a subformula is unlikely to be true.  We
%could mention it in this paper.  (I'm hesitant because ICDT reviewers
%are conservative theoreticians and wont appreciate this potential
%application.) This is something worth pursuing in the future}
%
%\guy{From an email by Dan: (1) Is the ``untyped cycle'' $\exists x, y, z E(x,y) \lor E(y,z) \lor E(z,x)$ reducible to C3?
%(2) Is ``non-transitivity''  $\exists x, y, z E(x,y) \lor E(y,z) \lor
%\neg E(x,z)$ reducible to C3?}

%% file: sharpp1.tex
\section{The Three Lemmas}

\subsection{Proof of Lemma~\ref{lemma:exists}: Removing Exists} Following the
proof of~\cite{DBLP:conf/kr/BroeckMD14}, we show how to eliminate
existential quantifiers from a $\WFOMC$ problem (a form of
Skolemization).  Assume that $\Phi$ is in prenex normal form: $\Phi =
Q_1 x_1 Q_2 x_2 \ldots Q_k x_k \Psi$, where each $Q_i$ is either
$\forall$ or $\exists$, and $\Psi$ is quantifier-free.  Let $i$ be the
first position of an $\exists$, and denote $\varphi(\mathbf{x},x_i) =
Q_{i+1} x_{i+1} \ldots Q_k x_k \Psi$; note that $\varphi$ is a formula
with free variables $\mathbf{x} = (x_1, \ldots, x_{i-1})$.  We have:
\begin{align*}
  \Phi = & \forall \mathbf{x} \exists x_i \varphi(\mathbf{x},x_i)
\end{align*}
Let $A$ be a fresh relational symbol of arity $i$.  The new formula
$\Phi'$ is:
\begin{align}
  \Phi' = &  \forall \mathbf{x} ((\exists x_i \varphi(\mathbf{x},x_i)) \Rightarrow A(\mathbf{x}))\label{eq:elim:exists}
\end{align}

Let $w', \bar w'$ denote the weights of $\Phi$'s vocabulary extended
with $w(A)=1$ and $\bar w(A)=-1$.  We claim that $\WFOMC(\Phi, n, w, \bar
w) = \WFOMC(\Phi', n, w', \bar w')$.  Consider a possible world $D
\subseteq \Tup(n)$ that satisfies $\Phi'$.  Call $D$ ``good'' if it
also satisfies $\Phi$.  In a good world $D$, for any constants
$\mathbf{x} = \mathbf{a}$, the sentence $\exists x_i
\varphi(\mathbf{a}, x_i)$ is true, hence $A(\mathbf{a})$ is also true
(because $\Phi'$ is true), and therefore the weight of $D$ is the same
as the weight of $D-\set{A}$, the world obtained from $D$ by removing
all tuples referring to the relational symbol $A$: $W(D, w', \bar w')
= W(D-\set{A}, w, \bar w)$.  Thus, the sum of the weights of the good
worlds (see Eq.(\ref{eq:def:wmc})) is precisely $\WFOMC(\Phi, n, w, \bar
w)$.  We prove that the sum of the weights of the bad worlds is zero.
Let $D$ be a bad world: it satisfies $\Phi'$ but not $\Phi$.  Thus,
there exists some constants $\mathbf{a}$ s.t. the sentence $\exists
x_i \varphi(\mathbf{a}, x_i)$ is false; choose $\mathbf{a}$ to be the
first such constants, in some lexicographic order.  Let $D'$ be the
world obtained from $D$ by flipping the status of $A(\mathbf{a})$:
thus $D, D'$ are identical, but one sets $A(\mathbf{a})$ to true and
the other to false.  Both satisfy $\Phi'$, and $W(D, w', \bar w') = -
W(D', w', \bar w')$, therefore they cancel out in the sum of
Eq.(\ref{eq:def:wmc}).  This proves that $\WFOMC(\Phi, n, w, \bar w) =
\WFOMC(\Phi', n, w', \bar w')$.

We note that $\Phi'$ can be written equivalently as:
\begin{align*}
\Phi' = &   \forall \mathbf{x} \forall x_i (\neg \varphi(\mathbf{x},x_i) \vee A(\mathbf{x}))
\end{align*}
In other words, we have replaced the first existential quantifier in
$\Phi$ by a universal quantifier (and may have increased the number of
$\exists$ on positions $i+1, i+2, \ldots$).  Lemma \ref{lemma:exists}
follows by applying this procedure inductively.

\subsection{Proof of Lemma~\ref{lemma:neg}: Removing Negation} Let $\neg
\psi(\mathbf{x})$ be a negated subformula of $\Phi$, with $k$ free
variables $\mathbf{x}$.  Let $A$, $B$ be two new relational symbols of
arity $k$.  Let $\Phi_p$ denote the sentence obtained from $\Phi$ by
replacing the subformula $\neg \psi(\mathbf{x})$ with $A(\mathbf{x})$.
Denote:
\begin{align}
  \Delta = & \forall \mathbf{x} [(\psi(\mathbf{x}) \vee A(\mathbf{x}))  \wedge (A(\mathbf{x}) \vee B(\mathbf{x})) \wedge (\psi(\mathbf{x})  \vee B(\mathbf{x}))] \label{eq:delta}
\end{align}
Extend the weight functions $w, \bar w$ to $w', \bar w'$ by setting
$w(A) = \bar w(A) = w(B) = 1$, $\bar w(B) = -1$.  Define $\Phi' =
\Phi_p \wedge \Delta$.  We claim that $\WFOMC(\Phi, n, w, \bar w) =
\WFOMC(\Phi', n, w', \bar w')$.  To prove this, consider a world $D
\subseteq \Tup(n)$ over the vocabulary of $\Phi'$, and assume that $D$
satisfies $\Phi'$.  Call $D$ ``good'', if the statement $\forall
\mathbf{x} (\psi(\mathbf{x}) \text{ xor } A(\mathbf{x}))$ holds.  It
is easy to see that in any good world, $\forall \mathbf{x}
B(\mathbf{x})$ holds too, hence the good world has the same weight as
the world obtained by stripping it of the additional relations $A,B$,
and, furthermore, their contributions to $\WFOMC(\Phi', n, w', \bar
w')$ is precisely $\WFOMC(\Phi, n, w, \bar w)$.  Consider a bad world:
it satisfies $\Phi'$, but there exists $\mathbf{a}$ such both
$\psi(\mathbf{a})$ and $A(\mathbf{a})$ are true.  In that case
$B(\mathbf{a})$ can be set arbitrarily to true or false and still
satisfy the formula $\Phi'$, hence the contributions of these two
pairing worlds cancel out.  This proves that $\WFOMC(\Phi, n, w, \bar
w) = \WFOMC(\Phi', n, w', \bar w')$.

Lemma \ref{lemma:neg} follows by apply this process repeatedly.

\subsection{Proof of Lemma~\ref{lemma:eq}: Removing Equality} Let $E(x,y)$ be
a new predicate symbol, with weights $w(E)=z$ and $\bar w(E)=1$, where
$z$ is a real value to be determined below.  Define $\Phi_E$ to be
obtained from $\Phi$ by replacing every equality predicate $x=y$ with
$E(x,y)$, and define:
\begin{align*}
  \Phi' = & \Phi_E \wedge \forall x E(x,x)
\end{align*}
Consider the count $f(z) = \WFOMC(\Phi', n, w', \bar
w')$ as a function of $z$, where $w', \bar w'$ extend $w, \bar w$ with
$w(E)=z, \bar w(E)=1$.  This is a polynomial of degree $n^2$ in $z$.
Since $\Phi'$ asserts $\forall x E(x,x)$, all monomials in $f$ have a
degree $\geq n$.  Let $c\cdot z^n$ be the monomial of degree $n$.
Then we claim that its coefficient $c = \WFOMC(\Phi, n, w, \bar w)$.
Indeed, every world $D$ where $E$ has exactly $n$ tuples is a world
where $E$ is interpreted as the equality predicate.  We can compute
$c$ using $n+1$ calls to an oracle for $f(z)$ as follows.  Fix $\delta
> 0$, and denote $\Delta^0 f = f$, $\Delta^{k+1}f(z) = (\Delta^k
f)(z+\delta)-(\Delta^k f)(z)$.  Then $\Delta^n f(0) = c\cdot n! = {n
  \choose 0} f(0) - {n \choose 1} f(\delta) + {n \choose 2} f(2\delta)
- \cdots (-1)^n {n \choose n} f(n\delta)$.

%% The proof of Lemma \ref{lemma:eq} follows by noting that, once we
%% have an access to an oracle for $f(z)$ we can compute the
%% coefficient $c$ of the degree $n$ monomial using $n+1$ probes into
%% $f$: fix some real number $h > 0$ and define $\Delta^0 f(z) = f(z)$
%% and $\Delta^{i+1} f(z) = \Delta^i f(z+h) - \Delta^i f(z)$; then
%% $\Delta^n f(0) = n! c$ and this requires only $n+1$ probes into
%% $f$, namely for $f(0), f(h), f(2h), \ldots, f(nh)$.

\section{A \#P$_1$-hard Sentence~$\HardSentence$}

\label{app:encoding}

We prove here Lemma~\ref{lemma:2}: shows how to reduce a linear-time,
multi-tape counting TM with a unary input alphabet (such as the
\#P$_1$-complete TM) to the FOMC problem on a first-order
sentence. The sentence that encodes the \#P$_1$-complete TM is
referred to as $\HardSentence$.  This proof is based on the standard
encoding of a deterministic Turing machine into first-order logic, as
used to prove Trakhtenbrot's
theorem~\cite[p.~167]{DBLP:books/sp/Libkin04}.  We extend this
construction in several ways: (1) towards non-deterministic counting
Turing machines, (2) with multiple tapes, (3) with a run time of $c
\cdot n$ for some fixed $c$, instead of $n$, (4) to have $n$ symbols 1
on the input tape, followed by symbols 0, and finally (5) to obtain a
sentence in \FO{3}.

As discussed in Section~\ref{s:encodingddiscussion}, we need to encode run times and space with lengths up to $c \cdot n$, yet we only have a domain size of exactly $n$ available. This is solved by partitioning the run time into $c$ \emph{epochs} of $n$ steps, and the space into $c$ \emph{regions} of $n$ cells. 
Moreover, we assume w.l.o.g.\ that there are two symbols: $\{0,1\}$.

\subsection{Signature}

The signature of $\HardSentence$ consists of the following predicates $P/a$, where $a$ is the arity of $P$:
\begin{itemize}
  \item $<\!/2$, denoting a  \emph{strict linear order} on the domain,
  \item $\Succ/2$, denoting the \emph{successor} relation w.r.t.\ the order on the domain,
  \item $\Min/1$ and $\Max/1$, denoting the \emph{smallest} and \emph{largest} domain element
  \item \emph{state} predicates $S_{qe}/1$, where $S_{qe}(t)$ is true precisely when the machine is in state $q$ at time $t$ in epoch $e$,
  \item \emph{head} predicates $H_{\tau er}/2$, where $H_{\tau er}(t,p)$ is true precisely when at time $t$ in epoch $e$, the head for tape $\tau$ is at position $p$ in region $r$, and
  \item \emph{tape} predicates $T_{s\tau er}/2$, where $T_{s\tau er}(t,p)$ is true precisely when at time $t$ in epoch $e$, tape $\tau$ contains symbol $s \in \{0,1\}$ at position $p$ in region $r$,
  \item \emph{movement} predicates $\GoLeft_{\tau er}/2$ and $\GoRight_{\tau er}/2$, where $\GoLeft_{\tau er}(t,p)$ is true precisely when the head on tape $\tau$ at time $t$ in epoch $e$ is to the left of $p$ in region $r$ (or when $p,r$ is the first cell on its tape and the head is there), and $\GoRight$ is defined similarly, and 
  \item \emph{frame} predicate  $\Unchanged_{\tau er}/2$, where we have that $\Unchanged_{\tau er}(t,p)$ is true precisely when position $p$ in region $r$ of tape $\tau$ did not change going from time $t$ in epoch $e$ to the next time step.
\end{itemize}

\subsection{Sentences}

To encode the Turing machine, we let $\HardSentence$ consist of the following sentences.
\begin{enumerate}
  \item $<$ is an arbitrary strict linear order (total, antisymmetric, irreflexive, and transitive):
  \begin{align*}
    \forall x, \forall y,~& \neg (x=y) \Rightarrow (x < y) \lor (y < x) \\
    \forall x, \forall y,~& \neg (x < y) \lor \neg (y < x)\\
    \forall x, \forall y, \forall z,~& (x < y) \land (y < z) \Rightarrow (x < z)
  \end{align*}  
  \item $\Min$ is the smallest element, and $\Max$ is the largest element:
  \begin{align*}
    \forall x,~ \Min(x) & \Leftrightarrow  \neg \exists y,  (y < x)\\
    \forall x,~ \Max(x) & \Leftrightarrow  \neg \exists y,  (x < y)
  \end{align*}
  \item $\Succ$ is the successor relation:
  \begin{align*}
    \forall x, \forall y,~ \Succ(x,y) \Leftrightarrow (x < y) \land  \neg \exists z, (x < z) \land (z < y )
  \end{align*}
  \item At any time, the machine is in exactly one state:
  \begin{align*}
    \bigwedge_{q,q',e: q \neq q'} \forall t,& ~ \neg S_{qe}(t) \lor \neg S_{q'e}(t) \\
    \bigwedge_{e} \quad\,\, \forall t, & ~ \bigvee_q S_{qe}(t)
  \end{align*}
  \item At any time, the head is in exactly one position per tape:
  \begin{enumerate}
    \item The head is in at least one position:
    \begin{align*}    
      \bigwedge_{\tau,e} \forall t, \exists p, \bigvee_r ~& H_{\tau er}(t,p)
    \end{align*}
    \item The head is in at most one region:
    \begin{align*}    
      \bigwedge_{\tau,e,r} \forall t, \forall p,~& H_{\tau er}(t,p) \Rightarrow \bigwedge_{r': r' \neq r} \forall p', \neg H_{\tau er'}(t,p') 
    \end{align*}
    \item The head is in at most one position per region:
    \begin{align*}    
      &\bigwedge_{\tau,e,r} \forall t, \forall p,~ H_{\tau er}(t,p) \\
      &\quad \Rightarrow \neg \exists p', \neg (p = p') \land H_{\tau er}(t,p') 
    \end{align*}
  \end{enumerate}
  \item At any time, each tape position has exactly one symbol:
  \begin{align*}
    \bigwedge_{\tau,e,r} \forall t, \forall p,~& T_{0\tau er}(t,p) \Leftrightarrow \neg T_{1\tau er}(t,p)
  \end{align*}
  \item In the initial configuration of the TM (first time step), 
  \begin{enumerate}
    \item it is in state $q_1$, and its heads are in the first position:
  \begin{align*}
    \forall x, \Min(x) \Rightarrow S_{q_1e_1}(x) \land \bigwedge_\tau H_{\tau e_1r_1}(x, x) 
  \end{align*}
  \item the first (input) tape $\tau_1$ contains $n$ symbols 1 in the first region, followed by symbol 0 in all other regions (starting with cell $n+1$), and all other tapes $\tau_{i}$ contain symbol 0:
  \begin{align*}
    \forall t, \Min(t) \Rightarrow  &\forall p, T_{1\tau_1e_1r_1}(t,p) \land \bigwedge_{i:i>1} T_{0\tau_1e_1r_{i}}(t,p) \\
    & \land \bigwedge_{i,r: i > 1} T_{0\tau_{i}e_1r}(t,p)
  \end{align*}
  \end{enumerate}
  \item An encoding of the transition relation $\delta$. For example, that state $q_a$ operates on tape $\tau_a$, and that $\delta(q_a,0) = \{(q_b,1,L),(q_c,0,R)\}$ is encoded into the following sentences.
  \begin{enumerate}
    \item What changed when $t$ is before the end of an epoch (i.e., has a successor in the epoch):
    \begin{align*}
      &\bigwedge_{e,r} \forall t, t', \forall p,~ \left[ 
      \begin{array}{r}
        S_{q_ae}(t) \\ \land H_{\tau_aer}(t,p) \\ \land T_{0\tau_aer}(t,p) \\ \land \Succ(t,t')
      \end{array}
      \right] 
      \Rightarrow  \\
      & \quad \left[ 
      \begin{array}{r}
        S_{q_be}(t')  \\ \land \GoLeft_{\tau_a er}(t',p) \\\land T_{1\tau_aer}(t',p)
      \end{array}
      \right] \lor \left[ 
      \begin{array}{r}
        S_{q_ce}(t') \\ \land \GoRight_{\tau_a er}(t',p) \\\land T_{0\tau_aer}(t',p)
      \end{array}
      \right] 
    \end{align*}
    \item What changed when $t$ is at the end of an epoch:
    \begin{align*}
      &\bigwedge_{i,r: 1 \leq i < c} \!\!\!\! \forall t, t', \forall p, \left[ 
      \begin{array}{r}
        S_{q_ae_i}(t) \\ \land H_{\tau_ae_ir}(t,p) \\ \land T_{0\tau_ae_ir}(t,p) \\ \land \Max(t) \\ \land \Min(t')
      \end{array}
      \right] 
      \Rightarrow \\
      &\quad \left[ 
      \begin{array}{r}
        S_{q_be_{i+1}}(t')  \\ \land \GoLeft_{\tau_a e_{i+1}r}(t',p) \\\land T_{1\tau_ae_{i+1}r}(t',p)
      \end{array}
      \right] \lor \left[ 
      \begin{array}{r}
        S_{q_ce_{i+1}}(t') \\ \land \GoRight_{\tau_a e_{i+1}r}(t',p) \\\land T_{0\tau_ae_{i+1}r}(t',p)
      \end{array}
      \right] 
    \end{align*}
    \item What does not change on the tapes: other cells in the region of $\tau_a$ where the head is, regions with no head, and tapes other than~$\tau_a$.
  \begin{align*}
  & \bigwedge_{e,r} \forall t, \forall p, \left[ 
      \begin{array}{r}
        S_{q_ae}(t) \\ \land H_{\tau_aer}(t,p)
      \end{array}
      \right] \\
      & \quad \Rightarrow
      \forall p', (p = p') \lor \Unchanged_{\tau_a er}(t,p') \\
      & \bigwedge_{\tau,e,r} \forall t, \forall p, H_{\tau er}(t,p) \\
      & \quad \Rightarrow
      \bigwedge_{r':r' \neq r} \forall p, \Unchanged_{\tau er'}(t,p)\\
      & \bigwedge_{e} \forall t, S_{q_ae}(t) \\
      & \quad \Rightarrow
      \bigwedge_{\tau,r:\tau \neq \tau_a} \forall p, \Unchanged_{\tau er}(t,p)
  \end{align*}
    \item The positions of the heads on tapes other than $\tau_a$ do not change:
  \begin{align*}
      &\bigwedge_{\tau,e,r: \tau \neq \tau_a} \forall t,t', \forall p, 
      \left[ 
      \begin{array}{r}
        S_{q_ae}(t) \\ \land H_{\tau er}(t,p) \\ \land \Succ(t,t')
      \end{array}
      \right] \\
      & \quad \Rightarrow H_{\tau er}(t',p) \\
      &\bigwedge_{\tau,i,r: \tau \neq \tau_a, 1 \leq i < c} \forall t,t', \forall p, 
      \left[ 
      \begin{array}{r}
        S_{q_ae_i}(t) \\ \land H_{\tau e_i r}(t,p) \\ \land \Max(t) \\ \land \Min(t')
      \end{array}
      \right] \\
      &\quad \Rightarrow H_{\tau e_{i+1}r}(t',p)
  \end{align*}
  \end{enumerate}
  \item The movement predicates are defined as
    \begin{align*}
      \bigwedge_{\tau,e,r} \forall t, \forall p, p', 
      \left[ 
      \begin{array}{r}
        \GoLeft_{\tau er}(t,p) \\ \land \Succ(p',p)
      \end{array}
      \right]  
      &\Leftrightarrow  H_{\tau er}(t,p') \\
      \bigwedge_{\tau,e,i: 1 \leq i < c} \forall t, \forall p, p', 
      \left[ 
      \begin{array}{r}
        \GoLeft_{\tau er_{i+1}}(t,p) \\ \land \Min(p) \\ \land \Max(p')
      \end{array}
      \right]  
      &\Leftrightarrow  H_{\tau er_{i}}(t,p')\\
      \bigwedge_{\tau,e} \forall t, \forall p, 
      \left[ 
      \begin{array}{r}
        \GoLeft_{\tau er_{1}}(t,p) \\ \land \Min(p)
      \end{array}
      \right]  
      &\Leftrightarrow  H_{\tau er_{1}}(t,p) \\      
      \bigwedge_{\tau,e,r} \forall t, \forall p, p', 
      \left[ 
      \begin{array}{r}
        \GoRight_{\tau er}(t,p) \\ \land \Succ(p,p')
      \end{array}
      \right]  
      &\Leftrightarrow  H_{\tau er}(t,p') \\
      \bigwedge_{\tau,e,i: 1 \leq i < c} \forall t, \forall p, p', 
      \left[ 
      \begin{array}{r}
        \GoRight_{\tau er_{i}}(t,p) \\ \land \Max(p) \\ \land \Min(p')
      \end{array}
      \right]  
      &\Leftrightarrow  H_{\tau er_{i+1}}(t,p')\\
      \bigwedge_{\tau,e} \forall t, \forall p, 
      \left[ 
      \begin{array}{r}
        \GoRight_{\tau er_{c}}(t,p) \\ \land \Max(p)
      \end{array}
      \right]  
      &\Leftrightarrow  H_{\tau er_{c}}(t,p)
    \end{align*}
  \item The frame predicates are defined as
  \begin{align*}
      &\bigwedge_{s,\tau,e,r} \forall t,t', \forall p, 
      \left[ 
      \begin{array}{r}
        T_{s\tau er}(t,p) \\ \land \Unchanged_{\tau er}(t,p) \\ \land \Succ(t,t')
      \end{array}
      \right] \\ 
      & \quad \Leftrightarrow  T_{s\tau er}(t',p) \\
      \end{align*}
      \begin{align*}
      &\bigwedge_{s,\tau,i,r: 1 \leq i < c} \forall t,t', \forall p, 
      \left[ 
      \begin{array}{r}
        T_{s\tau e_{i}r}(t,p) \\ \land \Unchanged_{\tau e_{i}r}(t,p) \\ \land \Max(t) \\ \land \Min(t')
      \end{array}
      \right] \\ 
      & \quad \Leftrightarrow  T_{s\tau e_{i+1}r}(t',p)
  \end{align*}
  \item The machine terminates in an accepting state (e.g., $q_1$, $q_5$, $q_{42}$, etc.) :
  \begin{align*}
    \forall t, \Max(t) & \Rightarrow S_{q_1e_c}(t) \lor S_{q_5e_c}(t)\lor S_{q_{42}e_c}(t) \lor \dots 
  \end{align*} 
\end{enumerate}

It is easy to verify that $\HardSentence$ uses at most three logical variables per sentence, and that $\HardSentence$ is therefore in \FO{3}. Note that \FO{3} permits variables to be reused within the same sentence.

For a fixed model of $<\!/2$, that is, a fixed order on the domain, the models of $\HardSentence$ for domain size $n$ correspond one-to-one to the accepting computations of the TM on input~$n$. Since there are exactly $(n!)$ models of $<\!/2$, we can compute the number of accepting computations from the FOMC efficiently.

\section{PTIME Data Complexity for \FO{2}} \label{s:fo2-ptime-proof}
%  \dan{Guy: feel free to change the wording} 

The proof of the fact that the data complexity for \FO{2} is in PTIME is
spread over two references, ~\cite{broeck2011completeness} and
\cite{DBLP:conf/kr/BroeckMD14}. 
\summary{
A brief proof, for
completeness, can be found in the long version of this paper, which is available at \url{http://arxiv.org/abs/1412.1505}.
}{
We include here a brief proof, for
completeness.

Given an \FO{2} formula $\varphi$ of size $s$, we start by applying
the reduction in~\cite{DBLP:journals/bsl/GradelKV97}, which converts
$\varphi$ into a formula $\varphi^*$ with the following properties:

\begin{itemize}
\item Every relational symbol occurring in $\varphi^*$ has arity at
  most 2.
% \item $\varphi^*$ has size $O(s \log s)$.
\item Items \ref{item:scott:1} and \ref{item:scott:3} of Scott's
  reduction hold. (Item \ref{item:scott:2} becomes: $\varphi^*$ has
  size $O(s \log s)$.  In our case $\varphi$ is fixed, so it suffices
  to note that the size of $\varphi^*$ is $O(1)$.)
\end{itemize}

The reduction consists of Scott's reduction described above, plus the
following transformation that ensures that all relational symbols have
arity $\leq 2$.  Replace each relational atom of arity $>2$ by a new
unary or binary symbol, for example, replace the atoms
$R(x,y,x)$, $R(y,y,y)$, $R(x,x,y)$ by $R_1(x,y)$, $R_2(x)$, $R_3(x,y)$.  Then
append to $\varphi^*$ conjuncts asserting how the new relational
symbols relate, for example $\forall x (R_1(x,x)\leftrightarrow
R_2(x))$; we refer the reader to~\cite{DBLP:journals/bsl/GradelKV97}
for details.

We perform one more transformation: remove all existential quantifiers
by using Lemma~\ref{lemma:exists}, thus transforming $\varphi^*$ into
a universally quantified sentences:
\begin{align*}
  \varphi^* = \forall x\, \forall y\, \psi(x,y)
\end{align*}
where $\psi(x,y)$ is a quantifier-free formula.

If $\varphi^*$ contains any relational symbol $R$ of arity zero then
we perform a Shannon expansion and compute $P(\varphi^*) =
\Pr(\varphi^*[R=\texttt{false}]) \cdot (1-p(R)) +
\Pr(\varphi^*[R=\texttt{true}]) \cdot p(R)$.  Thus, we can assume
w.l.o.g.\ that all relational symbols in $\varphi^*$ have arity 1 or 2.

Assume first that all relational symbols in $\varphi^*$ have arity 2.
Then we write  its lineage as:
\begin{align*}
  F = \bigwedge_{a,b \in [n]: a < b} \psi(a,b) \wedge \bigwedge_{c \in [n]} \psi(c,c)
\end{align*}
Since all atoms are binary, for any two distinct sets $\set{a,b}\neq
\set{a',b'}$, the formulas $\psi(a,b)$ and $\psi(a',b')$ are
independent probabilistic events, because they depend on disjoint sets
of ground tuples: one depends on tuples of the form $R(a,b)$ or
$R(b,a)$, the other on tuples of the form $R(a',b')$ or $R(b',a')$,
and they are disjoint.  (This would fail if $\psi$ contained a unary
atom, say $U(x)$, because we may have $a=a'$, $b\neq b'$, and in that
case both formulas depend on the tuple $U(a)$.)  Therefore:
\begin{align*}
  \Pr(F) = \prod_{a,b \in [n]: a < b} \Pr(\psi(a,b)) \cdot \prod_{c \in [n]} \Pr(\psi(c,c))
\end{align*}
Because the probabilities are symmetric, the quantity $p_1 =
\Pr(\psi(a,b))$ is independent of $a,b$, while $p_2 = \Pr(\psi(c,c))$ is
independent of $c$, and both can be computed in time $O(1)$.  Therefore,
$\Pr(\varphi) = \Pr(\varphi^*) = p(F) = p_1^{n(n-1)/2} p_2^n$.  For a
simple illustration, consider $\varphi^* = \forall x \forall y (R(x,y)
\vee T(x,y)) \wedge (R(x,y) \vee T(y,x))$.  Then:
\begin{align*}
  F = & \bigwedge_{a,b \in [n]: a<b} (R(a,b) \vee T(a,b))\wedge(R(a,b) \vee T(b,a))\\
      & \qquad \qquad \wedge(R(b,a) \vee T(b,a))\wedge(R(b,a) \vee T(a,b))\\
  & \wedge \bigwedge_{c \in [n]} (R(c,c) \vee T(c,c))
\end{align*}
and the probability is given by $p_1^{n(n-1)/2} p_2^n$ where $p_1 =
\Pr((R(a,b) \vee T(a,b))\wedge(R(a,b) \vee T(b,a))\wedge(R(b,a) \vee
T(b,a))\wedge(R(b,a) \vee T(a,b)))$ and $p_2 = \Pr(R(c,c) \vee T(c,c))$,
both quantities that can be computed using brute force.

Next consider the case when $\varphi^*$ has both unary and binary
relational symbols.  Let $R_1, \ldots, R_m$ be all unary symbols.
Consider the $2^m$ cells defined by conjunctions of these atoms or
their negation, denote them $C_1, \ldots, C_{2^m}$; that is $C_1(x)
\equiv \neg R_1(x) \wedge \cdots \wedge \neg R_m(x)$, $\ldots$,
$C_{2^m}(x) \equiv R_1(x) \wedge \cdots \wedge R_m(x)$.  Let $P$
denote any partition of $[n]$ into $2^m$ disjoint sets, i.e. $P=(S_1,
\ldots, S_{2^m})$ such that $S_1 \cup \cdots \cup S_{2^m} = [n]$.
Denote $(C_1, \ldots, C_{2^m}) = P$ the event that the $2^m$ cells
define precisely the partition $P$.  Then, summing over all partitions
$P$ gives us:
\begin{align}
  \Pr(\varphi^*) = \sum_P \Pr(\varphi^* \wedge (C_1,\cdots, C_{2^m})=P) \label{eq:pcp}
\end{align}
Next, we split $\varphi^*$ into a conjunction of several formulas,
each $x$ ranging over some cell $S_i$ and $y$ over some cell $S_j$:
\begin{align*}
  \varphi^* = & \bigwedge_{i, j \in [2^m]: i<j} \forall x:S_i, \forall y: S_j, (\psi(x,y) \wedge \psi(y,x)) 
       \\
       &\wedge \bigwedge_{\ell \in [2^m]} \forall x:S_\ell, \forall y: S_\ell, \psi(x,y)
\end{align*}
When $x$ ranges over $S_i$, then every unary predicate $R(x)$
containing the variable $x$ is either true or false.  Similarly, when
$y$ ranges over $S_j$, a predicate $R(y)$ is either true or false.
Let $\psi_{ij}(x,y)$ (or $\psi_\ell(x,y)$) denote the formula
$\psi(x,y) \wedge \psi(y,x)$ (or $\psi(x,y)$) where all the unary
predicates have been replaced by true or false, according to the cells
$S_i, S_j$ (or $S_\ell$ respectively).  Therefore:
\begin{align*}
  \varphi^* = & \bigwedge_{i, j \in [2^m]: i<j} \forall x:S_i, \forall y: S_j, \psi_{ij}(x,y)
      \\
      & \wedge  \bigwedge_{\ell \in [2^m]} \forall x:S_\ell, \forall y: S_\ell, \psi_\ell(x,y)
\end{align*}
Notice that $\psi_{ij}$ and $\psi_{\ell}$ have only binary predicates.
All conjuncts in the expression above are independent probabilistic
events: if $\set{i_1,j_1} \neq \set{i_2,j_2}$ then $\forall x:S_{i_1},
\forall y: S_{j_1}, \psi_{i_1j_1}(x,y)$ and $\forall x:S_{i_2}, \forall
y: S_{j_2}, \psi_{i_2j_2}(x,y)$ are independent.  Therefore, denoting
$n_i = |S_i|$ for $i=1,2^m$, we have: $\Pr(\varphi^* \wedge (C_1,\cdots,
C_{2^m})=P)= \prod_{i, j \in [2^m]: i<j} q_{ij} \cdot \prod_{\ell \in
  [2^m} r_\ell$, where:
\begin{align*}
  q_{ij} = & \Pr(\forall x:S_i, \forall y: S_j, \psi_{ij}(x,y)) \\
         = & \prod_{a \in S_i, b \in S_j} \Pr(\psi_{ij}(a,b)) = r_{ij}^{n_in_j}\\
  r_\ell = & \Pr(\forall x:S_\ell, \forall y: S_\ell, \psi_\ell(x,y)) \\
         = & \prod_{a, b\in S_\ell: a<b} \Pr(\psi_\ell(a,b) \wedge \psi_\ell(b,a)) \cdot \prod_{c \in S_\ell} \Pr(\psi_\ell(c,c))\\
        = & s_\ell^{n_\ell*(n_\ell-1)/2} \cdot t_\ell^{n_\ell}
\end{align*}
where $r_{ij} = \Pr(\psi_{ij}(a,b))$, $s_\ell = \Pr(\psi_\ell(a,b) \wedge
\psi_\ell(b,a))$, and $t_\ell = \Pr(\psi_\ell(c,c))$ are independent of
the choices of $a, b, c$ respectively, and can be computed by brute
force in time $O(1)$.  Finally, we use the fact that the probabilities
are symmetric, which implies that the expression in Eq.~(\ref{eq:pcp})
depends only on the cell cardinalities $n_1, \ldots, n_{2^m}$, and not on
the actual cells $S_1, \ldots, S_{2^m}$.  Therefore:
\begin{align*}
  \Pr(\varphi^*) = & \sum_{n_1,\ldots, n_{2^m}: n_1+\cdots+n_{2^m}=1}\frac{n!}{n_1!\cdots n_{2^m}!} \\ & \qquad \cdot r_{ij}^{n_in_j}s_\ell^{n_\ell*(n_\ell-1)/2} \cdot t_\ell^{n_\ell}
\end{align*}
For a simple illustration, consider $\varphi^* = \forall x \forall y
(R(x) \vee U(x,y) \vee T(y))\wedge (\neg R(x) \vee \neg U(x,y) \vee
\neg T(y))$.  Denoting the four cells $\neg R \wedge \neg T$, $\neg R
\wedge T$, $R \wedge \neg T$, $R \wedge T$ by $C_1, \ldots C_4$
respectively, we split $\varphi$ into a conjunct of $6+4$ expressions,
such that in each expression $x$ and $y$ are restricted to the domains
$C_i$ and $C_j$ respectively, for $i\leq j$.  Denoting $n_1, \ldots,
n_4$ the sizes of these cells, we have:
\begin{align*}
  \Pr(\varphi^*) = \sum_{n_1+\cdots +n_4 = n} \frac{n!}{n_1!n_2!n_3!n_4!}r_{ij}^{n_in_j}s_\ell^{n_\ell*(n_\ell-1)/2} \cdot t_\ell^{n_\ell}
\end{align*}
where $r_{12} = \Pr(U(a,b))$ (because $\forall x : S_1, \forall y :
S_2, (R(x) \vee U(x,y) \vee T(y))\wedge (\neg R(x) \vee \neg U(x,y)
\vee \neg T(y)) \equiv \forall x : S_1, \forall y : S_2, T(x,y)$),
and similarly for the others.
}